\documentclass[prb,superscriptaddress,twocolumn,aps,nobibnotes,nofootinbib, amssymb,floatfix,a4paper]{revtex4-2}
\usepackage{amsmath,amsfonts,enumerate,hyperref,tikz,pgfplots,thmtools,amsthm, adjustbox}
\usepackage{qcircuit}

\usepackage[most]{tcolorbox}   
\usepackage{enumitem}          

\newtcolorbox{summarybox}[1][]{
  enhanced,
  breakable,
  colback=gray!3, colframe=gray!50,
  boxrule=0.4pt, arc=2pt,
  left=6pt,right=6pt,top=6pt,bottom=6pt,
  fonttitle=\bfseries,
  fontupper=\small,
  title=Workflow summary,
  #1
}

\hypersetup{colorlinks=true,linkcolor=blue,citecolor=blue,urlcolor=blue}

\newcommand{\bra}[1]{\langle {#1} |}
\newcommand{\ket}[1]{| {#1} \rangle}
 
\newcommand{\ketbra}[2]{\ensuremath{\left|#1\right\rangle\!\!\left\langle#2\right|}}
\newcommand{\braket}[2]{\ensuremath{\!\left\langle#1|#2\right\rangle}\!}

\newcommand{\tr}[1]{\mathrm{Tr}\left( #1 \right)}

\renewcommand{\v}[1]{\ensuremath{\boldsymbol #1}}

\theoremstyle{plain}
\newtheorem{thm}{Theorem}

\newtheorem{lemma}[thm]{Lemma}

\theoremstyle{definition}

\begin{document}
	
	\title{Randomized adiabatic quantum linear solver algorithm with optimal complexity scaling and detailed running costs}
	\author{David Jennings}
	\affiliation{PsiQuantum, 700 Hansen Way, Palo Alto, CA 94304, USA}
	\author{Matteo Lostaglio}
			\thanks{Lead author email: mlostaglio@psiquantum.com}
	\affiliation{PsiQuantum, 700 Hansen Way, Palo Alto, CA 94304, USA}

	\author{Sam Pallister}
		\affiliation{PsiQuantum, 700 Hansen Way, Palo Alto, CA 94304, USA}.
	\author{Andrew T Sornborger}
    \affiliation{Computer, Computational, and Statistical Sciences Division, Los Alamos National Laboratory, Los Alamos, New Mexico 87545, USA}
	\author{Yi\u{g}it Suba\c{s}\i}
	\affiliation{Computer, Computational, and Statistical Sciences Division, Los Alamos National Laboratory, Los Alamos, New Mexico 87545, USA}

	\begin{abstract}
		Solving linear systems of equations is a fundamental problem with a wide variety of applications across many fields of science, and there is increasing effort to develop quantum linear solver algorithms.  Ref.~\cite{subacsi2019quantum} proposed a randomized algorithm inspired by adiabatic quantum computing, based on a sequence of random Hamiltonian simulation steps, with suboptimal scaling in the condition number $\kappa$ of the linear system and the target error $\epsilon$. Here we go beyond these results in several ways. Firstly, using filtering~\cite{lin2020optimal} and Poissonization techniques~\cite{cunningham2024eigenpath}, the algorithm complexity is improved to the optimal scaling $O(\kappa \log(1/\epsilon))$ -- an exponential improvement in $\epsilon$, and a shaving of a $\log \kappa$ scaling factor in $\kappa$. Secondly, the algorithm is further modified to achieve constant factor improvements, which are vital as we progress towards hardware implementations on fault-tolerant devices. We introduce  a cheaper randomized walk operator method replacing Hamiltonian simulation -- which also removes the need for potentially challenging classical precomputations; randomized routines are sampled over optimized random variables; circuit constructions are improved. We obtain a closed formula rigorously upper bounding the expected number of times one needs to apply a block-encoding of the linear system matrix to output a quantum state encoding the solution to the linear system. The upper bound is $837 \kappa$ at $\epsilon=10^{-10}$ for Hermitian matrices.
	\end{abstract}
	
	\maketitle
	
	\renewcommand{\thefootnote}{\arabic{footnote}}
	\setcounter{footnote}{0}

    \section{Introduction}

	 Given an $N \times N$ matrix $A$ and a $N$-dimensional vector~$\v{b}$, quantum linear solver algorithms (QLSAs) are tasked with returning the solution of a linear system $A \v{y} = \v{b}$ encoded as a quantum state. Linear systems are ubiquitous, since many problems admit reductions to them. Notably, many quantum algorithms for linear \cite{berry2017quantum,childs2021high, krovi2022improved, berry2022quantum, ameri2023quantum, bagherimehrab2023fast}  and nonlinear \cite{liu2021efficient,an2022efficient, krovi2022improved, jin2022quantum, surana2022carleman, costa2023further, krovi2024quantum} differential equations rely on linear solvers. Other applications include data fitting~\cite{wiebe2012quantum}, scattering~\cite{clader2013preconditioned}, machine learning~\cite{rebentrost2014quantum, liu2024towards} and optimization~\cite{dalzell2022end, krovi2024quantum}.

	Formally, given $\epsilon>0$, a QLSA outputs a quantum state $\epsilon$-close to a vector $\ket{y} \propto A^{-1} \ket{b}$,  giving the solution vector $\v{y} = A^{-1}\v{b}$ encoded as a quantum state. 
	The cost of the algorithm is given in terms of its \emph{query complexity} $Q$, i.e., the number of times we need to apply unitaries (`oracles') for state preparation of $\ket{b}$ and a block-encoding of $A$, namely a unitary encoding $A/\alpha$ for some $\alpha >0$ in one of its blocks.  A worst-case upper bound for the running cost of the algorithm can be given as a function of three parameters: $(\epsilon, \kappa, \alpha)$, where $\kappa$ is an upper bound on the condition number of the matrix 
and $\alpha$ is a rescaling constant, which depends on the specific block-encoding construction. Running QLSAs on the early generations of fault-tolerant quantum computers will only be feasible if we are able to (1) Bring down their cost and (2) Successfully incorporate them within an end-to-end quantum algorithm. The present manuscript focuses on the first problem. We introduce a new randomized QLSA and formally prove its performance is competitive with state-of-the-art.  
      
After Harrow, Hassidim, and Lloyd (HHL) developed the first QLSA~\cite{harrow2009quantum}, an extensive literature focused on bringing down the \emph{asymptotic scaling} of $Q$. The original HHL algorithm had an $O(\kappa^2/\epsilon)$ complexity. The dependence on the condition number was almost quadratically improved to $O((\kappa/\epsilon^3) \textrm{polylog}(\kappa/\epsilon))$ by Ambainis, at the price of a worse error scaling~\cite{ambainis2012variable}. The poor dependence on $\epsilon$ was remedied by Childs \emph{et al.}~\cite{childs2017quantum}, who obtained an $O(\kappa \, \textrm{polylog}(\kappa/\epsilon))$ algorithm, and a similar scaling was also realized in Ref.~\cite{chakraborty2018power}. A drawback of these works is that they involve a complex `variable-time amplitude-amplification' routine. Using a technique inspired by adiabatic quantum computing, Suba\c{s}\i \emph{ et al.}~\cite{subacsi2019quantum} introduced an adiabatic randomized algorithm that removed the need for such a routine, but with scaling $O(\kappa \log(\kappa) / \epsilon)$, which, however, is exponentially worse in $\epsilon$ than Ref.~\cite{childs2017quantum}.  Lin \emph{et al.} introduced a filtering technique~\cite{lin2020optimal}, providing an alternative path to an exponentially improved error scaling. This technique was exploited in the QLSA by Costa \emph{et al.}~\cite{costa2022optimal}, to achieve a claimed optimal $O(\kappa \log (1/\epsilon))$ scaling, with a worst-case constant prefactor to $\kappa$ upper bounded by $2\times 10^5$. Numerical studies of the average-case constant prefactors on randomized low-dimensional and low condition number instances ($N \leq 16$, $\kappa \leq 50)$ shows the average-case constant prefactors are much lower than the worst-case upper bound~\cite{costa2023discrete}. These studies cannot be scaled to $(N,\kappa)$ relevant in applications, so we need to rely on extrapolation.

Recently, Dalzell~\cite{dalzell2024shortcut} proposed an elegant QLSA, that achieves an optimal scaling of $O(\kappa \log (1/\epsilon))$ with much stronger guarantees on the worst-case prefactor to $\kappa$, upper bounded by $80$. This is the current state-of-the-art for quantum resource estimates with rigorous cost guarantees. The picture has been recently complicated by the introduction of novel techniques that allow improvement in the complexity in the number of calls to the unitary preparing $\ket{b}$ at the price of a small overhead on the number of calls to the block-encoding of $A$, with constant prefactor analysis yet to be carried out~\cite{low2024quantum}. See~\cite{morales2024quantum} for references to a growing literature on the topic.
      
In this work we focus on showing that one can obtain state-of-the art rigorous performance guarantees for randomized adiabatic QLSA. Starting from the randomized adiabatic algorithm in \cite{subacsi2019quantum}, we apply `Poissonization'~\cite{cunningham2024eigenpath} with an optimized schedule, filtering techniques~\cite{lin2020optimal} and a novel randomized walk operator method that entirely foregoes the need for a Hamiltonian simulation subroutine. Combining these and smaller optimizations at the level of block-encoding constructions, the complexity scaling is improved from the $O(\kappa/\epsilon \log \kappa)$ of the original proposal~\cite{subacsi2019quantum} to the optimal $O(\kappa \log (1/\epsilon))$, with constant prefactors of about $1734$ ($867$ for Hermitian matrices), a factor of $10$ -- $20$ from the best available bounds for any QLSA~\cite{dalzell2024shortcut}. This puts adiabatic randomized methods on the map of the most competitive QLSA with rigorous non-asymptotic cost guarantees. In terms of the broader context and timeliness of our work, we would note the following. Firstly, we expect further improvements can be made on the adiabatic line by exploiting more refined techniques.  Secondly, a concrete virtue of our method over prior results is its simplicity of implementation -- for example, we do not require the computation of any Quantum Singular Value Transform (QSVT) phase factors. This could make it well-suited to near/mid-term implementations on actual quantum hardware. Thirdly, while there are now a few optimal linear solvers, which can be compared in the worst-case via the constant pre-factors, this simple comparison need not hold for particular classes of problems. Ultimately it will be important to compare different linear solvers on important problems of interest, and one method could be more easily tailored to a specific problem at hand (similar to Hamiltonian simulation via Trotterization versus QSVT). Given the elementary ingredients of our algorithm, we expect that it will be well-suited to further optimization and application. In summary, given the very small handful of optimal linear solvers in existence, we expect that going forward the core method developed here has the ability to both compete with and complement existing methods.

\subsection*{ Overview of contributions} 
	
The high-level 
goal
of fault-tolerant adiabatic algorithms is to encode the normalized solution $\ket{y}$ of the linear system into the nullspace of a Hamiltonian $H(1)$. To obtain $\ket{y}$, we prepare an eigenstate of zero eigenvalue of a simple Hamiltonian $H(0)$ and then change the Hamiltonian to $H(1)$ along a discrete trajectory $H(s_j)$, where each $H(s_j)$ is constructed from the data $A$ and~$\ket{b}$. The scheduling is constructed so that one has rigorous guarantees on the quality of the output.

The original algorithm~\cite{subacsi2019quantum} fixes a set of $s_j \in [0,1]$, for $j=1,\dots, q$, and at each point it performs a Hamiltonian simulation routine for a randomized time $t_j$. The time $t_j$ is sampled from a uniform distribution over an interval wide enough that it approximatively dephases the system with respect to the eigenbasis of $H(s_j)$, up to a controlled error. This sequence of dephasings probabilistically maps the zero eigenstate of $H(s_{j-1})$ into that of $H(s_{j})$ at each step $j=1,\dots q$. Roughly speaking, the minimal gap of the sequence of Hamiltonians is $O(\kappa)$, leading to $t_j = O(\kappa)$, and $q=O(\epsilon^{-1} \log \kappa )$. Overall, this leads to the $O(\epsilon^{-1} \kappa \log \kappa)$ complexity scaling reported in~\cite{subacsi2019quantum}.

We modify the original algorithm~\cite{subacsi2019quantum} in a number of ways: 
\begin{enumerate}
    \item We generate $s_j$ in accordance with a Poisson process with a rate depending on the gap as proposed and discussed in Ref.~\cite{cunningham2024eigenpath}, which leads to $O(\log \kappa)$ in savings. We optimize the rate and tighten the error analysis to achieve constant factor savings.
    \item Inspired by previous works~\cite{boixo2009eigenpath, poulin2018quantum}, we show that we can entirely forego the Hamiltonian simulation routines.
    Dephasing is instead achieved by applying a walk operator $W(s_j)$, that is constructed from a block-encoding of $H(s_j)$, a random number of times $m$.
    This simplifies the quantum algorithm and reduces its cost. It also removes the need for classical precomputations of phase angles at each $j$, which is required by state-of-the-art Hamiltonian simulation techniques based on quantum signal processing~\cite{gilyen2018quantum}.
    \item Using results from eigenpath traversal theory~\cite{boixo2009eigenpath}, we sample $m$ at each $s_j$ from an optimized non-uniform random time variable, leading to  constant factor savings.
    \item We prepare a state with a constant error $\epsilon = 1/2$ and then apply filtering techniques~\cite{lin2020optimal, costa2022optimal}, which probabilistically prepares a state close to the target with exponential savings in $\epsilon$, from $O(1/\epsilon)$ to $O(\log (1/\epsilon))$. 
\end{enumerate} 

 Another difference from Ref.~\cite{subacsi2019quantum} is that we perform a detailed cost analysis, providing analytical, worst-case non-asymptotic bounds on the number, $Q$, of applications of the core unitary block-encodings given $\alpha$, $\kappa$ and $\epsilon$. The following is our main result:
		\begin{thm}[Optimal QLSA with explicit counts]
		\label{thm:querycountsQLSA}
		Consider a system of linear equations $A\v{y}=\v{b}$, where $A$ is an $N \times N$ dimensional matrix  scaled so that the singular values of $A$ lie in $[1/\kappa, 1]$.  Denote by $\ket{b}$ the normalized state that is proportional to $\sum_i b_i\ket{i}$, and by $\ket{A^{-1}b}$ the normalized state proportional to $A^{-1}\ket{b}$.  Assume access to
		\begin{enumerate}
\item A  unitary $U_A$ that encodes the matrix $A/\alpha$ in its top-left block, for some constant $\alpha \ge 1$, using a number of ancilla qubits equal to $a$,
\item An oracle, $U_b$, preparing $\ket{b}$.
\end{enumerate}
Then, there is a randomized quantum algorithm that outputs a quantum state $\epsilon$-close in $1$--norm to $\ket{A^{-1}b}$, using, in the worst case, an expected number of calls $Q^*$ to (controlled) $U_{A}$ or $U_{A}^\dag$ and $2 Q^*$  to (controlled) $U_b$ or $U_b^\dag$, where $Q^*=O(\kappa \log 1/\epsilon)$. Specifically,
		\footnotesize
		\begin{align*}
		 Q^* \le 835.4 \alpha \kappa + \alpha \kappa \ln \frac{2}{\sqrt{1+\epsilon/4}-1} +3.
		\end{align*}
		\normalsize
		The success probability is lower bounded by \mbox{$1/2 - \epsilon/4$}. The algorithm requires $a+ 7 + \lceil \log_2 N\rceil $ logical qubits. A factor of $2$ to the cost and an ancilla qubit can be saved if $A$ is Hermitian.
         	The expected query complexity including the failure probability is
	\begin{equation}
 \label{Eq:querybound}
		Q = 2Q^*/(1- \epsilon/2).
	\end{equation}
	\end{thm} 

 We highlight that, since the algorithm has a randomized component, for a fixed instance the cost of the algorithm is a random variable. We upper bound the average cost taking the worst case over all $N \times N$ matrices with fixed condition number upper bound, $\kappa$.  Alternatively, since the single-shot failure probability is $1/2+\epsilon/4 \approx 1/2$, if we run the simulation a number $\sim \log_2 (1/\delta)$ times, we reduce the failure probability to any $\delta >0$ and obtain a query count $\lceil Q^* \log_2 (1/\delta) \rceil $.  The parameter $\alpha$ is common to all algorithms invoking a block-encoding of the linear system matrix, and
its value depends on the underlying access model to $A$, whose choice should be tailored to each specific problem. 
 A number of constructions exist for sparse matrices~\cite{lin2022lecture, camps2022explicit, sunderhauf2023block}, which have $\alpha$ equal to the geometric mean of column and row sparsity,  and furthermore have $\mathrm{polylog}(N)  \log(1/\epsilon)$ gate cost if the position of the nonzero elements and their values are efficiently computable~\cite{lin2022lecture, camps2022explicit, sunderhauf2023block}.  Under these assumptions $a = O(\mathrm{polylog}(N))$. Other constructions allow to efficiently block-encode dense matrices with special structure~\cite{nguyen2022block, li2023efficient}. Finally, note that the above cost does not include the extraction of application-specific relevant information from the solution vector.

 \begin{table*}[t!]
    \centering
    \begin{adjustbox}{max width=\textwidth}
\begin{tabular}{|c|c|c|c|}
    \hline
     \textbf{Quantum linear solver} & \textbf{Asymptotic} & \textbf{Explicit upper bound} &  \textbf{Phase angle}\\  & \textbf{complexity} & at $(\alpha,\kappa,\epsilon)=(1, 10^6, 10^{-10})$ & \textbf{computation} \\
    \hline
    \hline
     Randomized adiabatic method \cite{subacsi2019quantum} & $O(\kappa \log(\kappa) /\epsilon)$ & $1.03 \times 10^{17}$ ($5.14 \times 10^{16}$) & Yes \\
    Deterministic adiabatic walk method \cite{costa2022optimal} & $O(\kappa \log(1/\epsilon))$ & $2.35 \times 10^{11}$ ($1.17 \times 10^{11}$) & No \\
    QSVT reflection and projection~\cite{dalzell2024shortcut} & $O(\kappa \log(1/\epsilon))$ & $8.02 \times 10^7$ ($8.02 \times 10^7$) & Yes \\
    \hline
    \hline
     Randomized adiabatic walk method & $O(\kappa\log(1/\epsilon))$ & $1.72 \times 10^9$ ($8.61 \times 10^8$) & No  \\ 
    \hline
\end{tabular}
    \end{adjustbox}
    \caption{Comparison of the randomized adiabatic walk method proposed in this work with alternative  optimal quantum linear solvers for which explicit, rigorous query complexity bounds have been derived. The values in parenthesis in the third column refer to the special case in which the linear system of equation is Hermitian. The practical performance of all these methods can be expected to be significantly better than what the upper bounds suggest. An initial numerical study of the deterministic adiabatic method and a prior non-optimal version of the present algorithm has been done in ~\cite{costa2023discrete}. In future work, it would be of interest to perform a detailed numerical study of the present optimal algorithm, and compare it with the other asymptotically optimal algorithms on a range of benchmark problems. }
\end{table*}

\subsection*{Overview of the algorithm}
	
	We now provide a description of the core components of the algorithm and leave the detailed analysis to the rest of the paper. A summary of the workflow is given at the end of the section.
	
	\bigskip
	
	\textbf{Step 1: Setting up the problem and constructing the oracles.} A linear system $A \v{y} = \v{b}$ is given, together with an upper bound, $\kappa$, on the condition number of $A$ and an error tolerance, $\epsilon$. Via an appropriate rescaling of the linear system, discussed in Sec.~\ref{app:hermitianrescaling},
	we can always take $A$ to be an $N \times N$ non-singular
	matrix with $\|A\| \leq 1$,  where $\|A\|$ is the operator norm of $A$.
	
	We assume access to two unitaries (the `oracles') that are at the basis of QLSA:
	\begin{enumerate}
		\item A unitary $
		U_b \ket{0} := \ket{b} \propto \sum_{j=1}^N b_j \ket{j}$, encoding $\v{b}$ as a quantum state when applied to a reference state~$\ket{0}$. 
		\item A unitary block-encoding of $A$, i.e. a unitary matrix $U_A$ with the block form
		\begin{equation}
			U_A=	\begin{bmatrix}
				A/\alpha & * \\
				* & *
			\end{bmatrix},
		\end{equation}
		where $\alpha \geq 1$ without loss of generality and stars indicate any additional matrices consistent with unitarity. 
	\end{enumerate}
	More formally, an $(\alpha, a, \epsilon_{M})$--\emph{block-encoding} of a $2^n \times 2^n$ matrix $M$ 
    is a $2^{a+n} \times 2^{a+n}$ unitary  $U_M$ such that
	\begin{equation}
		U_M \ket{0^a} \ket{\psi_n} =  \ket{0^a} \frac{M'}{\alpha}\ket{\psi_n} + \ket{\perp},
	\end{equation}
	where $(\bra{0^a} \otimes I_n)\ket{\perp} = 0$ and $\|M - M'\|\leq \epsilon_{M}$. Here $\ket{\psi_n}$ denotes an arbitrary $n$-qubit state, $\ket{0^a}:=\ket{0}^{\otimes a}$, and $I_n$ is the identity over $n$ qubits. In other words, $U_M$ encodes in a block (identified by the first $a$ qubits being in state zero) a matrix proportional to $M'$, which is $\epsilon_{M}$-close to~$M$. 
	
	The central goal of a QLSA is to output a sufficiently good approximation to the quantum state $\ket{y}$, where
	\begin{equation}
 \label{eq:definitionofy}
		\ket{y} \propto A^{-1} \ket{b}.
	\end{equation}
	The difficulty of the problem is quantified in terms of the \emph{query complexity} $Q$, i.e., the number of times we need to implement (controlled) $U_A$ or $U^\dag_A$, $U_b$ or $U^\dag_b$ in order to realize this output. In this work we shall be problem-agnostic, and assume access to an $(\alpha,a,0)$--block-encoding of the linear system matrix $A$,  although the extension to $\epsilon_{A}>0$ is straightforward.

	\bigskip
	
	\textbf{Step 2: Hamiltonian encoding and Poisson adiabatic trajectory.}  Extending the approach from Ref.~\cite{subacsi2019quantum}, we drop the assumption that $A$ is Hermitian and encode the solution of the \emph{Hermitian extension} of the linear system as an eigenstate of zero eigenvalue of a particular Hamiltonian.  Specifically, we introduce	
    \small\begin{equation}\label{eqn:adiabatic-H}
		H(s) = \ketbra{0}{1} \otimes A(s) \Pi + \ketbra{1}{0} \otimes \Pi A(s), \quad s \in [0,1],
	\end{equation}
    \normalsize
where we have, 
    \begin{align}
		\label{eq:As}
        \Pi &:= I - \ketbra{+,0,b}{+,0,b} \nonumber \\
		A(s) &:= (1-s) Z \otimes I  + s X \otimes \bar{A},
	\end{align}
   and 
   \begin{align}
   \label{eq:Abardef}
       \bar{A} = \ketbra{0}{1} \otimes A + \ketbra{1}{0} \otimes A^\dag
   \end{align} 
   is the Hermitian extension of $A$. The linear system associated to $\bar{A}$ is
		\begin{equation}
  \label{eq:Hermitianextendedlinearsystem}
		\bar{A}
		 \ket{1}\ket{y}	 =  	\ket{0} \ket{b}.
	\end{equation}
 Since $\bar{A}^{-1} = \ketbra{1}{0} \otimes A^{-1} + \ketbra{0}{1} \otimes (A^\dag)^{-1}$, using Eq.~\eqref{eq:definitionofy} the solution to the extended linear system is 
	\begin{equation}
		\label{eq:finalsystemsolution}
		\bar{A}^{-1} \ket{0} \ket{b} =  	\ket{1} A^{-1} \ket{b}  \propto \ket{1, y}.
	\end{equation}
If $N=2^n$, then $H(s)$ defines a one-parameter family of Hamiltonians on $n+3$ qubits. Setting 
\begin{align}
    \ket{y(s)} \propto  \ket{0} \otimes A(s)^{-1} \ket{+,0,b},
\end{align} 
we have that the family of states $\{\ket{y(s)}\}$ are eigenstates in the nullspace of $H(s)$:
	\begin{equation}
		H(s) \ket{y(s)} = 0, \quad \textrm{for all} \; s \in [0,1].
	\end{equation}
	Moreover, for every $s$ the nullspace of $H(s)$ is $2$--dimensional, with $\{\ket{y(s)}, \ket{1, +,0, b}\}$ providing an orthonormal basis.
	The input $\ket{b}$ and normalized solution $\ket{y}$ to the linear system are encoded in eigenstates of zero eigenvalue of $H(s)$ via 
 \begin{align}
 \label{eq:initialandfinalstates}
     \ket{y(0)} = \ket{0,-, 0, b}, \quad \ket{y(1)} = \ket{0,+, 1, y}.
 \end{align} 

The intuition from the quantum adiabatic theorem is that if we start from a preparation of $\ket{y(0)}$ and evolve under the Hamiltonian $H(s)$ while changing the parameter $s$ sufficiently slowly, we output a good approximation to $\ket{y(1)}$.  This works even if the nullspace of $H(s)$ is a $2$-dimensional space spanned by $\{\ket{y(s)}, \ket{1,+,0,b}\}$, because the evolution under $H(s)$ does not cause any transition between the two orthogonal eigenstates in the nullspace.
	
More precisely, the adiabatic protocol proceeds in discrete steps as follows: we fix $\gamma \in (0,1)$, where $1-\gamma$ encodes the fidelity to which the adiabatic protocol tries to prepare $\ket{y(1)}$.  Given $\kappa$ and $\gamma$ and a small interval $d s$ around $s$, we generate a `dephasing event' at $s$ with a probability $\lambda(s) d s$, i.e., according to a Poisson process whose rate $\lambda(s)$ is larger  where the Hamiltonian gap is smaller. We take the general form~\cite{cunningham2024eigenpath} 
    \begin{align}
        \lambda (s) = \frac{C(\gamma)}{ \Delta(s)^{q} \Delta_{\mathrm{min}}^{1-q}},
        \label{eq:Poissonrate}
    \end{align} where $\Delta (s)$ is a lower bound on the gap between the zero and nonzero energies of $H(s)$ and $\Delta_{\mathrm{min}} = \min_{s \in [0,1]} \Delta(s)$. This is different from Ref.~\cite{subacsi2019quantum}, where a set of $s_j$ was predetermined as a function of $\kappa, \gamma$. We optimize the schedule and find that $q=1/2$ gives the best upper bounds, so we make this choice from now on.
    We also find that  $C=68.6$ suffices for $\gamma =1/2$. We can take\footnote{Note that in \cite{subacsi2019quantum} our $\Delta(s)$ is denoted by $\sqrt{\Delta(s)}$ instead.}
	\begin{align}
    \nonumber
		\Delta(s) & = \sqrt{(1-s)^2 + (s/\kappa)^2}, \quad
            \Delta_{\mathrm{min}} & =
 (1+\kappa^2)^{-1/2}.	
\end{align}  
For detailed derivations, see Sec.~\ref{app:groundspace}.

\bigskip
\textbf{Step 3: Dephasing events via randomized walk method.} For each $s_j$ associated to a dephasing event, we want to realize a dephasing in the eigenbasis of $H(s_j)$. More formally, the aim is to effect a quantum channel~\cite{boixo2009eigenpath}
	\begin{equation}
		\label{eq:Pj+1}
		\mathcal{P}_{j}(\rho) = P(s_j) \rho P(s_j) + \mathcal{E}_{j} \circ (I- P(s_j)) \rho (I- P(s_j)),
	\end{equation}
	where $P(s_j)$ is the projector onto the nullspace of $H(s_{j})$ and $\mathcal{E}_j$ is a channel mapping the nonzero eigenspaces onto themselves. $\mathcal{P}_j$ acts as a non-selective measurement in the eigenbasis of $H(s_j)$,  evolving the instantaneous eigenstate from $\ket{y(s_{j-1})}$ to $\ket{y(s_j)}$ with sufficiently high probability. 
    
The original proposal was to  perform at each $s_j$ a randomized Hamiltonian simulation for a time $t_j$ sampled from a uniform random variable~\cite{subacsi2019quantum}. However, this realizes a channel of the form~\eqref{eq:Pj+1} only approximately and hence leads to overheads that require further analysis~\cite{chiang2014improved}. This can be avoided as follows. We can instead sample $t_j$ according to a probability distribution $p_{s_j}(t)$, and correspondingly perform Hamiltonian simulation for time $\alpha_{s_j} t_j$, where  $\alpha_{s_j}$ is the block-encoding rescaling factor for $H(s_j)$, which we shall discuss later. A channel of the form~\eqref{eq:Pj+1} is then realized up to an error given by the maximum of the characteristic function of $p_{s_j}(t)$ evaluated at the nonzero energy eigenvalues of $H(s_j)$~\cite{boixo2009eigenpath}.  Therefore, if one uses a $p_{s_j}(t)$ that is bandwidth limited to $[-\Delta(s_j), \Delta(s_j)]$, the error vanishes and we realize a channel of the form $\mathcal{P}_j$ exactly. We need to find a probability distribution with this property that minimizes $\langle t_j \rangle$.\footnote{One could also consider jointly optimizing some combination of $\langle t_j \rangle$ and $\langle t_j^2 \rangle$, to put a penalty on excessive values of the variance. } Luckily, this problem has been studied before. Ref.~\cite{boixo2009eigenpath} proposed $p_{s_j}(t) \propto \textrm{sinc}^4(\Delta(s_j)t/4)$.  Ref.~\cite{sanders2020compilation}, in the context of combinatorial optimization algorithms, proposed a numerically optimized polynomial ansatz for the characteristic function of a probability distribution satisfying the properties we need, which returned a 46th-order polynomial (pages 35-36). 

Here, we simplify matters by explicitly constructing the probability distribution   \begin{equation}
		\label{eq:randomtime}
		p_{s_j}(t) \propto \frac{1}{\Delta(s_j)} \left(\frac{ J_{r}\left(\Delta(s_j) |t|/2\right)}{\Delta^{r-1}(s_j) |t|^r}\right)^2,
	\end{equation} 
	where $J_r(z)$ is the Bessel function of first kind of order $r=1.165$. This is obtained from an optimized quadratic order characteristic function, and has a value of $\langle t_j \rangle$ within $0.0022\%$ of the proposal in  Ref.~\cite{sanders2020compilation}. It has $\langle |t_j| \rangle = 2.32132/\Delta(s_j)$, and its variance is $9.36238/(\Delta(s_j))^2$. 

This leaves us with the requirement to perform, at each $s_j$, Hamiltonian simulation for a time $\alpha_{s_j} t_j$, with $t_j$ sampled from Eq.~\eqref{eq:randomtime}. The asymptotically best algorithms for Hamiltonian simulation are based on Quantum Signal Processing (QSP)~\cite{gilyen2018quantum} or Generalized QSP (GQSP)~\cite{motlagh2023generalized}. One drawback of this approach is that it requires us to compile phase factors for each $s_j$, involving potentially challenging classical computations. A second drawback is that the technique brings a constant factor overhead of at least $e/2$ in the query count from rigorous non-asymptotic estimates of the overhead in the degree of the Jacobi-Anger polynomial approximation of the complex exponential~\cite{gilyen2018quantum}.

   Here, instead, we replace Hamiltonian simulation by a random walk method, which removes both drawbacks. First, in Sec.~\ref{app:proofblockencodinghamiltonian} we present a construction realizing a self-inverse $(\alpha_s, a+2,0)$--block-encoding of $H(s)$ with $\alpha_s = (1-s) + \alpha s$ with a single call to $U_A$, using the Linear Combination of Unitaries (LCU) technique. The LCU construction of~Ref.~\cite{subacsi2019quantum} has $\alpha_s = 2(\alpha+1)$.  Since $\alpha_s $ enters linearly in $Q$, our construction cuts costs by about a factor of $2$. Then, we construct the walk operator
   \begin{align}
   \label{eq:walkoperator}
       W(s) = U_{H(s)} \mathcal{Z} U_{H(s)} \mathcal{Z},
   \end{align}
   where $\mathcal{Z} = (2 \ketbra{0^{a+2}}{0^{a+2}} - I)\otimes I$.
   This unitary can be written as $W(s) = e^{iH_W(s)}$, for some Hermitian, $H_W(s)$. The target adiabatic protocol is now realized by dephasing relative to the eigenspaces of $H_W(s)$, of which the $\pi$--eigenspace encodes the nullspace of $H(s)$. The operator $H_W(s)$ has a gap $\tilde{\Delta}_W(s)$ between $\pi$ and the remaining eigenvalues, which obeys $\tilde{\Delta}_W(s) \geq 2 \Delta(s)/\alpha_s$, as we show in Sec.~\ref{sec:randomizedwalkmethod}.

At each $s_j$, we shall apply $W(s)$ a number of times $m$ sampled from the probability distribution
  \begin{align}
        \label{eq:pjintegeresmain}
            p_{s_j}(m) \propto \frac{1}{\tilde{\Delta}_W(s_j)} \left( \frac{J_r(\tilde{\Delta}_W(s_j) |m|/2)}{\tilde{\Delta}_W^{r-1}(s_j) |m|^r} \right)^2, 
        \end{align}
        defined over the set $\mathbb{Z}$, and where $r = 1.165$.

We now have all the ingredients. We select $s_j$ according to the `Poisson random' process (with rate in Eq.~\eqref{eq:Poissonrate}), and for each $s_j$ we apply the walk operator $W(s_j)$ a number of times $m_j$ with probability in Eq.~\eqref{eq:pjintegeresmain}. For any given realization, we apply 
\begin{equation*}
W^{m_q}(s_q) W^{m_{q-1}}(s_{q-1}) \cdots W^{m_1}(s_1),
\end{equation*} 
on the input state. Averaging over the random variables we get a density operator $\rho(1)$ with 
	\begin{equation}
	\label{eq:idealrandomizedmethod}
		\bra{y(1)} \rho(1) \ket{y(1)} \geq 1/2.
	\end{equation}
    
	\bigskip
	\textbf{Step 4: Projection on the correct solution.} Next, we take the output of the adiabatic protocol and apply an $O(\epsilon)$-approximate projection onto the nullspace of $H(1)$, with a query cost $O(\kappa \log(1/\epsilon))$.

    Here we leverage the filtering results of Ref.~\cite{costa2022optimal}, which build upon the results in~\cite{lin2020optimal}. These results ensure that we can apply a block-encoding of $P(1)$, the projector onto the nullspace of $H(1)$ (the Hamiltonian at the end of the trajectory), 
	to the state $\rho(1)$ with a circuit involving on average $\left \lceil \frac{\alpha \kappa}{2} \log\left(\frac{2}{\epsilon_P}\right) \right \rceil$ applications of a block-encoding of $H(1)/\alpha$ or its inverse, and $2$ extra qubits. Since we have access to an $(\alpha, a+2,0)$ block-encoding of $H(1)$, this gives us an $(\alpha, a+4, \epsilon_P)$ block-encoding of $P(1)$. No complex phase factor pre-computations are required~\cite{costa2022optimal}. Counting $3$ extra qubits that we introduced in the constructions, overall we have used $7$ ancilla qubits.
 
An appropriate $\epsilon_P$ needs to be chosen by combining the errors from the adiabatic stage with those of the filtering stage and the corresponding success probabilities. We do this in Sec.~\ref{app:adiabaticfilteringerrors}. 
With this choice, if we succeed we output an $\epsilon$-approximation to the solution vector $\ket{y(1)}$ (and so $\ket{y}$).  Otherwise, we go back to Step 3 and repeat. 

    This completes an overview of the algorithm. We now move to a detailed analysis.

    \begin{summarybox}[label={workflow}] 
    {\bf Initialization}
\begin{enumerate}[leftmargin=*, itemsep=2pt, topsep=2pt]
\item Rescale the problem so that $\|A\|\leq 1$.
  \item Input circuit block-encoding of $A/\alpha$ for some $\alpha>1$.
  \item Input state preparation circuit for $\ket{b}$.
  \item Input error tolerance $\epsilon >0$.
  \item Input condition number upper bound $\kappa$.
  \end{enumerate}
   {\bf Hamiltonian encoding}
      \begin{enumerate}[leftmargin=*, itemsep=2pt, topsep=2pt]
  \item Define Hamiltonian encoding $H(s)$ as in Eq.~\eqref{eqn:adiabatic-H}, with block-encoding $U_{H(s)}$ constructed in Sec.~\ref{app:proofblockencodinghamiltonian}.
  \item Generate Poisson points $s_j$ in $[0,1]$ with rate $\lambda(s)$ according to Eq.~\eqref{eq:Poissonrate} with $\gamma =1/2$, $q=1/2$.
   \end{enumerate}
      {\bf Quantum algorithm}
     \begin{enumerate}[leftmargin=*, itemsep=2pt, topsep=2pt]
       \item Prepare $\ket{y(0)} = \ket{0,-,0,b}$ (Eq.~\eqref{eq:initialandfinalstates}).
       \item \emph{Quantum random walk:} for each Poisson point $s_j$, apply walk operator $W(s)$ in Eq.~\eqref{eq:walkoperator} a number of times $m$ sampled according to $p_{s_j}(m)$ in Eq.~\eqref{eq:pjintegeresmain}.
      \item  \emph{Filtering}: Apply quantum circuit plus measurement effecting a projection onto the nullspace of $H(1)$ as described in Sec.~\ref{app:filteringcost}.
      \item Repeat till success.
  \end{enumerate}

\end{summarybox}

\section{Algorithmic details}

	\subsection{Bringing the problem into standard form}
	\label{app:hermitianrescaling}

    \subsubsection{Rescaling}
	
	The central aim is to solve $A\v{y} = \v{b}$ in its coherent formulation $A\ket{y} \propto \ket{b}$. 
	It is useful to embed the linear system in another one for which the norm of the matrix of  coefficients is bounded by~$1$.
	
	Here we briefly discuss how to do so and why it does not affect the results of our work.   Let $N_A$ be an upper bound to the norm of $A$. Consider the rescaling $A' = A/N_A$. Clearly, $\| A' \| \leq 1$ and so all its singular values lie in the interval $[1/\|(A')^{-1}\|,1]$. We define $\kappa'=  \|(A')^{-1}\|$ so that all singular values of the rescaled matrix are included in $[1/\kappa' ,1]$. The parameter $\kappa'$ is an upper bound on the condition number of the rescaled matrix. 
	
	 In this paper, we assume access to an $(\alpha',a,0)$-block-encoding of $A'$. Note that without loss of generality we can assume $\alpha' \geq 1$. In fact, we can reabsorb $\alpha'<1$   into a redefinition of the normalization constant $N_A$.  Let us see why: Assume that the previous construction leads to a rescaled $A'$ with a block-encoding $U_{A'}$ with $\alpha' < 1$. This implies that $\|A'\|$ is strictly less than~1. Then define $A'' =  A'/\alpha'$.	The singular values of the	$A''$ matrix lie in the interval $[1/\kappa'',1]$, where $\kappa'' = \alpha' \kappa'$. In other words, the condition number upper bound of  $A''$ is a factor of $\alpha'$ smaller than that of $A'$. Furthermore, access to an $(\alpha',a,0)$-block-encoding of $A'$ is equivalent to access to an $(\alpha'',a,0)$-block-encoding of $A''$, where $\alpha''= 1$. 
     
     In other words, given $U_{A'}$ a block-encoding of $A'$, with a block-encoding rescaling parameter and condition number upper bound equal to $
	(\alpha', \kappa')$ respectively, then we have $U_{A'}= U_{A''}$
	a block-encoding for $A''$, with new parameters $
		(1, \alpha' \kappa')$.
	Since the query cost of QLSA is linear in the block-encoding rescaling and at least linear in the condition number, solving the linear problem for $A''$ is no more costly than solving the one for $A'$. Hence, this shows that we can take without loss of generality the rescaling factor of the block-encoding to be larger or equal to $1$. 

From now on, for simplicity of notation, we drop primes and assume the rescaling has been done. We have access to $U_A$ which is an $(\alpha,a,0)$ block-encoding of $A$, where $\alpha \geq 1$ and $\| A\| \leq 1$.

    \subsubsection{Hermitian extension}
    
We now consider the Hermitian extended linear system as in Eq.~\eqref{eq:Abardef}.
 We have $\| \bar{A} \| \leq 1$ and so all the singular values of $\bar{A}$ are between $[1/\bar{\kappa},1]$, where $\bar{\kappa}:=  \|\bar{A}^{-1}\|$. Furthermore, $\bar{\kappa} = \|\bar{A}^{-1}\| = \|A^{-1}\| = \kappa$, where we used that $A$ is invertible. Hence, the Hermitian extension does not change the bounds on the singular values. What does change is that  we need to block-encode $\bar{A}$ rather than $A$. As we shall discuss in more detail later, this can be done with at most one call to a block-encoding of $A$ and one call to a block-encoding of~$A^\dag$. If $A$ is Hermitian, the Hermitian extension is not required and a factor of $2$ is saved from the overall cost.

	\subsection{Properties of $H(s)$}
	\label{app:groundspace}

	For the results that follow we refer to  
 Ref.~\cite{subacsi2019quantum}. For clarity and ease we report what we need here.
	
	\subsubsection{The nullspace of $H(s)$} 
    We consider the family of Hamiltonians
	\begin{equation}
		H(s) = \ketbra{0}{1} \otimes A(s) \Pi + \ketbra{1}{0} \otimes \Pi A(s),
	\end{equation}
	with
	\begin{align}
             \Pi &= I - \ketbra{+,0,b}{+,0,b} \nonumber\\
		A(s) &= (1-s) Z \otimes I  + s X \otimes \bar{A},
	\end{align}
    where $Z$ is the Pauli $z$-matrix.
	We claim that for all $s$ the nullspace of $H(s)$ is $2$--dimensional and spanned by the orthonormal states $\ket{y(s)} \propto  \ket{0} \otimes A(s)^{-1} \ket{+,0,b}$ and $\ket{1,+,0,b}$. Let us prove that this is the case.
	
	Since $H(s)$ are Hermitian, the eigenvalues of $H(s)^2$ are exactly the square of the eigenvalues of $H(s)$ and they have the same eigenvectors. Hence consider
    \small \begin{align}
		H(s)^2 &= \ketbra{0}{0} \otimes A(s)\Pi A(s) + \ketbra{1}{1} \otimes \Pi A(s)^2 \Pi.
	\end{align}
    \normalsize
	From the above, the nullspace of $H(s)$ is spanned by the eigenvectors with zero eigenvalues of the two blocks of $H(s)^2$.
	
	The $\ket{0}$ block is spanned by vectors of the form $\ket{0} \otimes \ket{\psi}$. For this to be an eigenstate with eigenvalue zero we need
	\begin{equation}
    \label{eq:nullespacepropertiesderivation}
		(A(s)\Pi A(s)) \ket{\psi} = 0.
	\end{equation}

Let us show that $A(s)$ is invertible since $A$ is, by assumption, invertible. Compute

\begin{align}
\nonumber
    A^\dagger(s) A(s) &= (1-s)^2 I \otimes I + s^2 (I \otimes \bar{A}^\dag \bar{A}) \\ &=
     I \otimes I - s^2 ( I \otimes I - I \otimes \bar{A}^\dag \bar{A}),
\end{align} 
where we have used the facts that $\bar{A}$ is Hermitian and Pauli operators anticommute. 
Using Weyl's inequality, we can lower bound the smallest eigenvalue of $A^\dagger(s) A(s)$ with $(1-s)^2 + (s/\kappa)^2 \ge 1/\kappa^2 >0 $. Hence, $A(s)$ is invertible.

It follows that for every $\ket{\psi}$ there is a vector $\ket{\phi}$ such that $\ket{\psi} = A^{-1}(s) \ket{\phi}$. Then Eq.~\eqref{eq:nullespacepropertiesderivation} becomes
	\begin{equation}
		A(s) \Pi \ket{\phi} = 0,
	\end{equation}
	which holds if and only if $\ket{\phi} = \ket{+,0,b}$. Therefore, for all $s$ we have that $\ket{y(s)} \propto \ket{0}\otimes A^{-1}(s) \ket{+,0,b}$ is the unique eigenstate of zero eigenvalue of the top block. 
	
	We now consider the bottom block. The eigenstates with zero eigenvalue are of the form $\ket{1} \otimes \ket{\psi}$, where $\Pi A(s)^2 \Pi \ket{\psi} = 0$. Choosing $\ket{\psi} = \ket{+,0,b}$ gives an eigenstate with zero eigenvalue. Again, since $A(s)$ is invertible, we have that this state is unique.  In conclusion the nullspace of $H(s)$ is spanned by $\ket{y(s)} \propto \ket{0} \otimes A(s)^{-1} \ket{+,0,b}$ and $\ket{1,+,0,b}$. 
	
	\subsubsection{Solution encoding} 
    When $s=1$,  we have that $A(1)^{-1} = X \otimes \bar{A}^{-1}$ and so the associated nullspace vector is given by
	\begin{align}
		\ket{y(1)} \propto \ket{0} \otimes  (X \otimes \bar{A}^{-1}) \ket{+,0,b} \propto \ket{0, +, 1,y}.
	\end{align}
The state $\ket{y} \propto A^{-1} \ket{b}$ is the desired solution to the linear system.
	\subsubsection{Gap from zero}
    
 Now we consider the gap between the zero eigenvalue and the rest of the eigenvalues in $H(s)$, and show it is lower bounded by 
 \begin{align}
     \Delta(s) = \sqrt{(1-s)^2 + (s/\kappa)^2}.
 \end{align} 
 The gap between the zero eigenvalue and non-zero eigenvalues of $H(s)$ is simply the square root of the gap between the zero eigenvalue and non-zero eigenvalues for $H(s)^2$. If we set $B(s) = A(s) \Pi$, the top block of $H(s)^2$ reads $B(s) B(s)^\dag$ and the bottom block reads $B(s)^\dag B(s)$. The two blocks hence have the same spectrum and we can then focus only on the top block. It follows that the gap from zero of $H(s)$ is just the square root of the gap from zero of $A(s) \Pi A(s)$.
	
	The gap of $A(s) \Pi A(s)$ was analyzed in the Supplemental Material of Ref.~\cite{subacsi2019quantum} using Weyl's inequalities. One rewrites
	\begin{equation}
		A(s)\Pi A(s) = A(s)^2 - A(s) \Pi^\perp A(s),
	\end{equation}
	where $\Pi^\perp = \ketbra{+,0,b}{+,0,b}$.
	From Weyl's inequality the second smallest eigenvalue of $A(s)\Pi A(s) = A(s)^2 - A(s)\Pi^\perp A(s)$ (which is the gap from zero, since \mbox{$A(s)\Pi A(s)$} has a zero eigenvalue and is a nonnegative matrix) is lower bounded by the smallest eigenvalue of $A(s)^2$ plus the second smallest eigenvalue of $-A(s) \Pi^\perp A(s)$. We consider each separately.
	
	Note that $A(s)$ has eigenvalues $\pm \sqrt{(1-s)^2 + (s \lambda)^2}$, with $\lambda$ eigenvalues of $\bar{A}$. Recall that $A(s)$ is Hermitian, so the eigenvalues of $A(s)^2$ are of the form $(1-s)^2 + (s \lambda)^2$. Since $1/\kappa$ is a lower bound on the smallest singular value of $A$ by definition of $\kappa$, the smallest eigenvalue of $A(s)^2$ is lower bounded by $(1-s)^2 + (s/\kappa)^2$.

	Now consider $-A(s) \Pi^\perp A(s)$. This is a non-positive, rank-1 matrix, and so the second smallest eigenvalue is zero. We conclude that the gap of $A(s) \Pi A(s)$ is lower bounded by $(1-s)^2 + (s/\kappa)^2$, and so the gap from zero of $H(s)$ is lower bounded by $\Delta(s) = \sqrt{(1-s)^2 + (s/\kappa)^2}$. 
	
	\subsubsection{No transitions between orthogonal components of the nullspace.}
	We also claim that for any fixed $s$ the Hamiltonian  never couples the two eigenstates of zero eigenvalue. Specifically, we have that for every $s,s' \in [0,1]$
	\small
	\begin{align}
		&\bra{1} \bra{+,0,b} H(s) \ket{y(s')}  \nonumber \\ 
	 &	\hspace{-0.3cm} \propto \bra{1} \bra{+,0,b} (\ketbra{0}{1} \otimes A(s) \Pi + \ketbra{1}{0} \Pi A(s) ) \ket{0}  A^{-1}(s') \ket{+,0,b}\nonumber \\
	&	= \bra{+,0,b} \Pi A(s)A^{-1}(s') \ket{+,0,b} =0.
	\end{align}
	\normalsize
	Therefore, if we initialize the system in the state $\ket{y(0)}$ and perform the sequence of dephasings in the eigenbasis of the Hamiltonian $H(s)$, we will not generate a component along $\ket{1} \otimes \ket{+,0,b}$.

		\subsection{Block-encoding $H(s)$}
	\label{app:proofblockencodinghamiltonian}
	
	To construct the block-encoding of $H(s)$ from that of $A$ we rely on the linear combination of unitaries (LCU) technique. LCU allows us to obtain a $\left(\| \beta\|_1, \lceil \log_2 k \rceil , 0\right)$--block-encoding of $\sum_{i=1}^k \beta_i U_i$, where  $\|\beta\|_1= \sum_i \beta_i$, $U_i$ are unitaries and without loss of generality we assume $\beta_i > 0 $ by absorbing a phase in the definition of $U_i$~\cite{gilyen2018quantum, lin2022lecture}. LCU assumes access to the controlled unitary
	\begin{equation}
		V_S= \sum_{i=1}^k \ketbra{i}{i} \otimes U_i, \quad \quad  \textrm{\emph{``Select''}}
	\end{equation}
	and the state preparation unitary
	\begin{equation}
		V_P\ket{0^{{ \lceil \log_2 k \rceil}} }= \frac{1}{\sqrt{\|\beta\|_1}} \sum_{i=1}^k \sqrt{\beta_i} \ket{i}, \quad  \quad  \textrm{\emph{``Prepare''}},
	\end{equation}
	where we use the shorthand notation $\ket{0^a} :=\ket{0}^{\otimes a}$.
	The claimed block-encoding is then realized as 
	\begin{equation}
		W = (V_P^\dag  \otimes I)V_S (V_P \otimes I).
	\end{equation}
	
	Another general result concerns products of block-encodings. If $U_M$ is an $(\alpha_M, a_M, 0)$--block-encoding of $M$ and $U_N$ is an $(\alpha_N, a_N, 0)$--block-encoding of $N$, then $(I_{a_M} \otimes U_M)(I_{a_N} \otimes U_N)$ is an $(\alpha_M \alpha_N, a_M+a_N, 0)$--block-encoding of $MN$ (Ref.~\cite{gilyen2018quantum}, Sec. 4.4). 
	
	Using these results we can show the following, which as discussed improves the LCU decomposition of Ref.~\cite{subacsi2019quantum} by extending the discussion to a setting where $A$ is not necessarily Hermitian and sparse, while cutting the $1$-norm by at least a factor of $2$.
	
	\begin{lemma}[Constructing a block-encoding of $H(s)$]
		\label{res:blockencoding}
		Suppose we access an $(\alpha,a,0)$--block-encoding $U_A$ of $A$ and the state preparation unitary $U_b$ with 	$U_b \ket{0} = \ket{b}$.  We can then construct a unitary $U_{H(s)}$ that is an $$(\alpha_s= 1-s + \alpha s,a+2,0)$$
        block-encoding of 
			the Hamiltonian
            \begin{equation}
                H(s) = \ketbra{0}{1} \otimes A(s) \Pi + \ketbra{1}{0} \otimes \Pi A(s),
            \end{equation}
            where we have
		\begin{align*}
            \Pi &= I - \ketbra{+,0,b}{+,0,b} \\
			A(s) &= (1-s) Z \otimes I_{n+1}  + s X \otimes \bar{A}, \\
            \bar{A} &=  \ketbra{0}{1} \otimes A + \ketbra{1}{0} \otimes A^\dag,
		\end{align*}
 with $2$ calls to (controlled) $U_A$ or $U^\dag_A$ and $4$ calls to $U_b$ or~$U_b^\dagger$. Furthermore, the block-encoding is Hermitian, i.e., $U_{H(s)}^\dag = U_{H(s)}$. 
		 If $A$ is a $2^n\times 2^n$ matrix then $U_{H(s)}$ acts on $n+a+5$ qubits.
	\end{lemma}
	\begin{proof}
	 Given a block-encoding $U_A$ of $A$, $U^\dag_A$ block-encodes $A^\dagger$. Then
			\begin{equation}
				U_{\bar{A}} = (\ketbra{0}{0} \otimes U_A + \ketbra{1}{1} \otimes U_A^\dag) (X \otimes I),
			\end{equation}
			acts on $n+a+1$ qubits and
		gives an $(\alpha, a,0)$-block-encoding of $\bar{A}$ involving a single call to controlled $U_A$ and controlled $U_A^\dag$. Note that $U_{\bar{A}}$ is Hermitian. 
		
		Next, we
use LCU to obtain a block-encoding of $A(s)$ from the block-encoding of $\bar{A}$.  We do this by considering the expression
\begin{equation}\label{eqn:LCU-for-As}
(1-s) Z \otimes I_{n+a+1} + \alpha s X \otimes U_{\bar{A}},
\end{equation}
which defines an $n+a+2$ qubit operator which is a sum of unitaries and encodes $A(s)$ in the $\ket{0^a}$ block, with zero error.  This operator is block-encoded via LCU.
LCU requires a `select' unitary
	\begin{equation}
			V_S = \ketbra{0}{0} \otimes (Z \otimes I_{n+a+1}) + \ketbra{1}{1} \otimes (X \otimes  U_{\bar{A}}),
		\end{equation}
and a `prepare' unitary 
\begin{equation}
			V_P(s) \ket{0} = \frac{1}{\sqrt{1-s+\alpha s}} \left(\sqrt{1-s} \ket{0} + \sqrt{\alpha s } \ket{1}\right).
		\end{equation}
This introduces a single extra qubit to condition on to realize this linear combination, and has a scaling factor given by
\begin{equation}
\| (1-s, \alpha s)\|_1 = 1-s + \alpha s.
\end{equation}
The LCU unitary is given as
	\begin{equation}
U_{A(s)}:=  (V^\dag_P(s) \otimes I_{n+a+2}) V_S (V_P(s) \otimes I_{n+a+2}), 
\end{equation}
which is a unitary acting on $n+a+3$ qubits in total and provides a $(1-s+\alpha s, a+1, 0)$ block-encoding of the $n+2$ qubit operator $A(s)$. 

To construct a block-encoding of $H(s)$ we can combine the above block-encoding for $A(s)$ with a block-encoding of $\Pi$. In particular, we can block-encode \mbox{$\Pi = I -  \ketbra{+,0,b}{+,0,b}$} by writing
		\begin{align*}
			\Pi & = (\textrm{Had} \otimes U_b) (I - \ketbra{0^{n+2}}{0^{n+2}}) (\textrm{Had} \otimes U^\dag_b) \\ &  = \frac{1}{2}\left(I +  
			(\textrm{Had} \otimes U_b) e^{i \pi \ketbra{0^{n+2}}{0^{n+2}}} (\textrm{Had} \otimes U^\dag_b)\right)
		\end{align*}
		as a linear combination of two unitaries 
        \begin{align*}
            U_0 = I, \quad U_1 = (\textrm{Had} \otimes U_b) e^{i \pi \ketbra{0^{n+1}}{0^{n+1}}} (\textrm{Had} \otimes U^\dag_b),
        \end{align*}
        where $\textrm{Had}$ denotes the Hadamard matrix on one qubit, and in the second equality we used $ \ketbra{0^{n+1}}{0^{n+1}} = (I-e^{i \pi \ketbra{0^{n+1}}{0^{n+1}}})/2$ . We can then define the prepare and select unitaries
		\begin{equation}
			\tilde{V}_P = \textrm{Had}, \quad	\tilde{V}_S = \ketbra{0}{0} \otimes U_0 + \ketbra{1}{1} \otimes U_1.
		\end{equation} 
		The unitary $\tilde{W} = (\tilde{V}^\dag_P \otimes I_{n+1}) \tilde{V}_S (\tilde{V}_P \otimes I_{n+1})$ is then a $(1,1,0)$--block-encoding of $\Pi$,  using one call to $U_b$ and one to its inverse, which introduces one additional qubit. Hence from multiplying the block-encodings, we have that $U_{A(s)} \tilde{W}$ ($\tilde{W} U_{A(s)}$, respectively) is a $(1-s+\alpha s,a+2,0)$--block-encoding of $A(s) \Pi$ ($\Pi A(s)$, respectively) acting on $n+a+4$ qubits. 
		
		As a final step, note that
		\begin{align*}
			H(s) & = \ketbra{0}{1} \otimes A(s) \Pi + \ketbra{1}{0} \otimes \Pi A(s) \nonumber \\ & = (X \otimes I) \left[\ketbra{1}{1} \otimes A(s) \Pi + \ketbra{0}{0} \otimes \Pi A(s)\right],
		\end{align*}
		for $s \in [0,1]$, and so 
		\begin{equation}
			U_{H(s)}:=(X \otimes I) \left[\ketbra{1}{1} \otimes U_{A(s)} \tilde{W} + \ketbra{0}{0} \otimes \tilde{W} U_{A(s)}\right]
   \nonumber
		\end{equation}
is a $(1-s+\alpha s,a+2,0)$--block-encoding of $H(s)$ acting on a total of $n+a+5$ qubits, which is consistent with $H(s)$ acting on $n+3$ qubits in total. This occurs since we require $a+1$ auxiliary qubits for a block-encoding of $A(s)$ and one additional auxiliary qubit for the block-encoding of the projector.   Let us rewrite
			\begin{align}
				\nonumber
				U_{H(s)}&:=(X \otimes I)\left[   \ketbra{0}{0} \otimes  \tilde{W} + \ketbra{1}{1}  \otimes I \right] U_{A(s)}  \\ & \times  \left[\ketbra{0}{0} \otimes I + \ketbra{1}{1} \otimes \tilde{W} \right],
			\end{align}
		which shows that $U_{H(s)}$ can be obtained via a single call to $U_{A(s)}$ (i.e., a single call to controlled  $U_{\bar{A}}$, which is $1$ call to controlled $U_A$ and $1$ to controlled $U^\dag_A$) and $2$ calls to $\tilde{W}$ (i.e., $4$ calls to $U_b$ or $U^\dag_b$). 
		
		Finally, $U_{H(s)}^\dag = U_{H(s)}$ follows immediately from the fact that $U_{A(s)}^\dag = U_{A(s)}$ and $\tilde{W}^\dag = \tilde{W}$.

		Since $H(s)$ acts on 3 more qubits than $A$, the block-encoding $U_{H(s)}$ acts on $n+a+ 5$ qubits in total.
	\end{proof}
	Given this basic building block, we move on to the details of the randomized part of the algorithm.

     \subsection{Construction of a random variable for the randomization scheme}
     \label{sec:randomvariableconstruction}

     For the randomized protocol we shall need to construct a random variable, $p_{s}(t)$, for any $s\in [0,1]$ such that 
     \begin{itemize}
		\item The distribution $p_{s}(t)$ is band-limited: more precisely, its Fourier transform $\chi(\omega)$ has support in $[-\Delta(s), \Delta(s)]$.
		\item The random variable $T=t$ has $\langle |T|\rangle$ minimized.
	\end{itemize}
To construct $p_{s}(t)$, recall the definitions of Fourier transform $\mathcal{F}$ and its inverse $\mathcal{F}^{-1}$,
		\begin{align}
			\mathcal{F}(f)(\omega) &= \int_{-\infty}^{+\infty} e^{-i \omega t} f(t) dt, \\
			\mathcal{F}^{-1}(g)(t) &= \frac{1}{2 \pi} \int_{-\infty}^{+\infty} e^{i \omega t} g (\omega) d\omega.
		\end{align}
		To construct such a $p_{s}(t)$ with characteristic function supported only in $[-\Delta(s),\Delta(s)]$ we use the following construction strategy. We fix an ansatz characteristic function $g_s(\omega)$ with $g_s(0) = 1$ and supported only in $[-\Delta(s)/2,\Delta(s)/2]$. Its inverse Fourier transform gives a function $f_s(t)$ which will not in general be positive and is hence not a valid probability density. However we can set $p_{s} \propto f_s^2$ which, using the convolution theorem, has characteristic function $\chi_s = g_s \ast g_s$, where $\ast$ labels convolution. Hence, $\chi_s$ is supported in $[-\Delta(s),\Delta(s)]$, $p_{s}$ is positive and it is a valid probability density once normalized (if it can be normalized). These are the properties we needed to apply the techniques in Ref.~\cite{boixo2009eigenpath}.

		Let us see this strategy in action. Set
		\begin{equation}
			g_s(\omega) = \left(1- \frac{4 \omega^2}{\Delta(s)^2}\right)^{r-1/2} h\left(\frac{\omega}{\Delta(s)}\right),
		\end{equation}
		where we choose $r= 1.165$, $h(x)$ is the step function in $[-1/2,1/2]$, i.e., $h(x) = 1$ if $|x| \leq 1/2$ and $h(x) = 0$ otherwise (which we constructed by taking a quadratic expression with support in $[-\Delta(s)/2, \Delta(s)/2]$ and satisfying the constraints $g_s(\omega=0) =1$, $g_s(\omega=\Delta(s)/2)=0$). The support of $g_s$ is $[-\Delta(s)/2,\Delta(s)/2]$. The inverse Fourier transform reads
		\begin{equation}
			\mathcal{F}^{-1}( g_s(\omega)) \propto \frac{ J_{r}\left(\Delta(s) |t|/2\right)}{\Delta(s)^{r-1} |t|^{r}} : =f_s(t).
		\end{equation}
		Note that $f_s(t)$ is not a valid probability distribution since it can be negative, so we take $p_{s}(t) = f_s(t)^2/N_s$. It is readily verified that
		\begin{equation}
			\int_\mathbb{R} dt p_{s}(t) = 1.
		\end{equation}
		with normalization $ 0.2379128\dots  \times \Delta(s)$
		and therefore $p_{s}(t)$ is a valid probability  density on the real line.  The associated characteristic function $\chi_s(\omega)$ is the Fourier transform of this distribution and using the convolution theorem we have
		\begin{align}
			\chi_s(\omega) &:=  \int_{-\infty}^{+\infty} e^{i \omega t} p_{s}(t) dt \propto \mathcal{F}(f_s^2)(\omega) \\ &=  \mathcal{F}(\mathcal{F}^{-1}(g_s)^2)(\omega) \propto (g_s \ast g_s)(\omega).
		\end{align}
		We therefore have that $\chi_s(\omega)$ only has support in the interval $[-\Delta(s), \Delta(s)]$,  and the above $p_{s}(t)$ is a valid probability density.  
        
		Technically to conclude $\sup_{\omega \in  [-\Delta(s), \Delta(s)]} |\chi_s(\omega) |= 0$ above we would need $\Delta(s)$ to be a \emph{strict} lower bound to the zero eigenvalue gap of $H(s)$. However, $\Delta(s)$ is only a lower bound to the gap. By continuity one can of course repeat the same construction as the above with a characteristic function $g$ with support in $\left[\frac{\Delta(s) + \delta}{2}, \frac{\Delta(s) - \delta}{2}\right]$ and then make $\delta$ arbitrarily small, getting the same results as presented here.

\subsection{Randomization method using the block-encoding walk operator}
\label{sec:randomizedwalkmethod}

In this section we introduce a novel randomization method based on the walk operator obtained from block-encoding of the Hamiltonian, as opposed to relying on time evolutions. Start from the block-encoding of $H(s)$:
\begin{align}
U_{H(s)}|0\rangle|\psi\rangle=|0\rangle \frac{H(s)}{\alpha_s}|\psi\rangle+|\tilde\perp(s)\rangle,
\end{align}
where $(\langle 0| \otimes \mathbb{I})|\tilde\perp(s)\rangle=0$ and $|\tilde\perp(s)\rangle$ is subnormalized. Here $\alpha_s= 1-s + \alpha s$ is the block-encoding prefactor we obtained in the construction of Lemma~\ref{res:blockencoding}. Let $\left|E_i(s)\right\rangle$ be an eigenstate of $H(s)$ with eigenvalue $E_i(s)$. Then
\small
\begin{align*}
U_{H(s)}|0\rangle\left|E_i(s)\right\rangle & =\frac{E_i(s)}{\alpha_s}|0\rangle\left|E_i(s)\right\rangle+\left|\tilde{\perp}_{i}(s)\right\rangle \\
& =\frac{E_i(s)}{\alpha_s}|0\rangle\left|E_i(s)\right\rangle+\sqrt{1-\left(\frac{E_i(s)}{ \alpha_s}\right)^{2}}\left|\perp_{i}(s)\right\rangle \\
&= \cos(\theta_i(s)) \ket{0} \ket{E_i(s)} + \sin(\theta_i(s)) \ket{\perp_i(s)},
\end{align*}
\normalsize
where $\ket{\perp_{i}(s)}$ is now normalized  and satisfies 
\begin{align}
\label{eq:orthogonalqubitization}
   ( \bra{0} \otimes I)\ket{\perp_i(s)} = 0
\end{align} 
Also, we have defined $\theta_i(s) = \arccos(E_i(s)/\alpha_s)$. Since $U_{H(s)}$ is Hermitian (Lemma~\ref{res:blockencoding}), the $\mathrm{span} \{ \ket{0} \ket{E_i(s)}, \ket{\perp_i(s)}\}$ is an invariant subspace for $U_{H(s)}$~\cite{lin2022lecture}.

Let $\mathcal{Z}=(2\ket{0}\bra{0}-I)\otimes I$. Then it can be shown that the `doubled' walk operator $W(s)=U_{H(s)} \mathcal{Z} U_{H(s)} \mathcal{Z}$ acts in 2-dimensional invariant spaces   $\mathcal{H}_i:=\mathrm{span} \{ \ket{0} \ket{E_i(s)}, \ket{\perp_i(s)}\}$ as 
\begin{align}
W(s)|_{\mathcal{H}_i} = \left(\begin{array}{cc}
\cos \left(2 \theta_{i}(s)\right) & -\sin \left(2 \theta_{i}(s)\right) \\
\sin \left(2 \theta_{i}(s)\right) & \cos \left(2 \theta_{i}(s)\right)
\end{array}\right).
\end{align}
 The linear system Hamiltonian $H(s)$ of Eq.~\eqref{eqn:adiabatic-H} has a special spectrum. Its eigenvalues come in pairs of equal magnitude but opposite sign, except for two degenerate eigenvalues that are zero. We will label the positive eigenvalues in an increasing order with positive integers $j$ and the corresponding negative eigenvalues with $-j$. As for the degenerate eigenvalues 0 we use the notation
$\ket{E_0^+(s)}:=\ket{y(s)}$ and $\ket{E_0^-}:=\ket{1}\ket{+,0,b}$ as basis vectors to span this subspace.

Note that $\ket{E_0^+(1)}$ encodes the solution to the linear systems problem, $\ket{E_0^+(s)}$ is the eigenpath we wish to follow and the orthogonal state  $\ket{E_0^-}$ has no $s$-dependence. We have $E^\pm_0(s)=0$ and so $\theta_0(s)=\arccos(0)=\pi/2$ and hence $W(s)$ acts as minus identity in both of the 2-dimensional subspaces associated with the zero eigenvalues of $H(s)$. 

This shows that the `doubled' walk operator leaves the the zero energy space (corresponding to $\theta_0 = \pi/2$) invariant. However, we also need to verify that no other energy eigenspace ends up being associated to the same eigenphase of the walk operator. This requires every other eigenphase of the walk operator to be separated from $2\theta_0=\pi$ by a finite phase gap.

Since $E_i(s) / \alpha_s \in [-1,1]$, $\theta_i(s)=\arccos(E_i(s) / \alpha_s) \in (-\pi,\pi]$. Also, for $i \neq 0$, $|E_i(s)|/\alpha_s >\Delta(s)/\alpha_s$, and so   
\begin{align}
    \theta_i(s) \in (0, \pi/2 - \delta(s)) \cup (\pi/2 + \delta(s), \pi), 
\end{align} 
where $\pi/2 \mp \delta(s) = \arccos(\pm \Delta(s)/\alpha_s)$, and so $\delta(s) = \pi/2 - \arccos(\Delta(s)/\alpha_s) \geq \Delta(s)/\alpha_s$. This implies that the spectrum of $W(s)$ around the phase $\pi$ has a gap of 
\begin{align}
    \Delta_{W}(s) = \pi - 2\arccos(\Delta(s)/\alpha_s) \geq 2 \Delta(s) / \alpha_s.
\end{align} 
 Note this is tight only if $\Delta(s) \ll \alpha_s$.
We now note that for any $s\in [0,1]$ we can write W(s) = $e^{iH_W(s)}$ for some (non-unique) Hermitian operator $H_W(s)$. For any such $H_W(s)$ we see that $\Delta_W(s)$ is the gap between the $\pi$--eigenvalue of the operator $H_W(s)$ and its other eigenvalues.

The $\pi$--eigenspace of $H_W(s)$ is  a 4-dimensional subspace $\mathcal{H}_W^i$ with basis
\begin{align}
    \ket{\eta_0(s)} &= \ket{0}\ket{E_0^+(s)}, \\
    \ket{\eta_1} &= \ket{0}\ket{E_0^-}, \\
    \ket{\eta_2(s)} &= \ket{\perp_0^+(s)}, \\
    \ket{\eta_3(s)} &= \ket{\perp_0^-}, 
\end{align}
and the walk operator can be written as 
\begin{align}
    W(s) & = \bigoplus_{i\neq 0} \left(\begin{array}{cc}
\cos \left(2 \theta_{i}(s)\right) & -\sin \left(2 \theta_{i}(s)\right) \\
\sin \left(2 \theta_{i}(s)\right) & \cos \left(2 \theta_{i}(s)\right)
\end{array}\right)_{\mathcal{H}_i} \nonumber \\ & - \sum_{r=0}^3 \ket{\eta_r(s)}\!\bra{\eta_r(s)},
\end{align}
where $\ket{\perp_0^+(s)}$ is the vector defined in Eq.~\eqref{eq:orthogonalqubitization} to be orthogonal to $\ket{E_0^+(s)}$, and similarly $\ket{\perp_0^-}$ for $\ket{E_0^-}$.

Each realization of the randomization method implements a unitary of the form
\begin{align}
    W(s_k)^{m_k}\cdots W(s_2)^{m_2}  W(s_1)^{m_1}, 
\end{align}
where $k$, $\{s_j\}_{j=1}^{k}$, and $\{m_j\}_{j=1}^k$ are random variable samples, as discussed below.  

We could at this point use results from Ref.~\cite{boixo2009eigenpath} for the randomization method, but we want to improve further on it, so we will add randomization over the choice of $s_j$ values via Ref.~\cite{cunningham2024eigenpath}. In Ref~\cite{cunningham2024eigenpath}, one considers a Poisson process with rate $\lambda(s)$. At each jump point, $s_j$, a randomized time-evolution with respect to the Hamiltonian $H(s)$ is applied, as originally proposed in Ref.~\cite{subacsi2019quantum}, with the evolution time, $t_j$, chosen from a probability density~$p_{s_j}(t)$.

Here, we instead put forward a new scheme that entirely removes the Hamiltonian simulation routine and directly applies the walk operator $W(s_j)^{m_j}$, where $m_j$ is sampled according to the \emph{discrete} probability
        \begin{align}
    \label{eq:pjintegeres}
            p_{s_j}(m) = \frac{1}{\mathcal{N}(s_j) \tilde{\Delta}_W(s_j)} \left( \frac{J_r(\tilde{\Delta}_W(s_j) |m|/2)}{\tilde{\Delta}_W^{r-1}(s_j) |m|^r} \right)^2, 
        \end{align}
        with $m\in \mathbb{Z}$ an integer, $\tilde{\Delta}_W(s_j)$ any strict lower bound on $\Delta_W(s_j)$, and $r = 1.165$.
Note that
\begin{equation}
    W(s_j)^{m_j} = e^{im_j H_W(s_j)},
\end{equation}
and so amounts to an evolution under the Hamiltonian $H_W(s_j)$ for an integer-valued time $m_j$. This then allows us to realize dephasing onto the $\pi$--eigenspace of $H_W(s)$ via the randomized method applied to the walk operator.

The characteristic functions of the continuous version of the above discrete probability (where $m$ is extended to $m \in \mathbb{R}$, and the function renormalized accordingly) has support contained in $(-\Delta_W(s_j), \Delta_W(s_j))$ (see Sec~\ref{sec:randomvariableconstruction}), with $\Delta_W(s_j) \leq \pi$. By Lemma~4 of Ref.~\cite{boixo2009eigenpath}, its restriction to the integers (which is the discrete probability in Eq.~\eqref{eq:pjintegeres} that we are sampling from) is a well-defined probability distribution with characteristic function contained in $(-\Delta_W(s_j), \Delta_W(s_j))$. Using Theorem~1 of the same paper, the randomized protocol described
effects a channel on system plus ancilla of the form
\begin{align}
    \rho \mapsto \bar{P}_W(s_j) \rho \bar{P}_W(s_j) + \bar{\mathcal{C}}_{s_j} \circ (I- \bar{P}_W(s_j)) \rho (I- \bar{P}_W(s_j)),
\end{align}
where $\bar{P}_W(s_j)$ is the projector onto the subspace $\mathcal{H}^j_{W}$, and $\bar{\mathcal{C}}_{s_j}$ is a channel that maps the set of states with support entirely in the orthogonal subspace to $\mathcal{H}^j_{W}$ into itself. The expected number of applications of the walk operator is
\begin{align}
    \sum_{m=-\infty}^{+\infty} |m| p_{s_j}(m) \leq 2.322/\tilde{\Delta}_W(s_j),
\end{align}
where the bound has been numerically computed. Each application of the walk operator (or its inverse $W(s)^{-1}$ in the case of negative integer powers) corresponds to $2$ applications of the block-encoding of $H(s)$, and so $4$ applications of the basic unitaries $U_A$ and $U^\dag_A$. Using $\Delta_W(s) \geq 2 \Delta(s)/\alpha(s)$ we can upper bound the expected query cost of a randomization as
\begin{align}
\label{eq:querycostats}
Q(s_j) & \leq 4\times 2.322/ \Delta_W(s_j) \nonumber \\
& \leq  2\times 2.322 \alpha(s_j)/\Delta(s_j).
\end{align}

We now show that such an operation cannot cause transitions from $\ket{\eta_0(s_j)}$ to $\ket{\eta_r (s_{j+1})}$ for $r=1,2,3$. First 
\begin{align*}
\bra{\eta_1(s_{j+1})}W(s_{j+1})\ket{\eta_0(s_j)} &= - \braket{\eta_1(s_{j+1})}{\eta_0(s_j)} \\
    &= - \braket{0}{0}\braket{E_0^-}{E_0^+(s_j)} \\
    &= 0.
\end{align*}
Similarly for the other two eigenstates we have
\begin{align*}
    \bra{\eta_{2,3}(s_{j+1})}W(s_{j+1})\ket{\eta_0(s_j)} &= - \braket{\eta_{2,3}(s_{j+1})}{\eta_0(s_j)} \\
    &= - \bra{\perp_0^\pm(s_{j+1})}(\ket{0}\ket{E_0^+(s_j)}) \\
    &= 0,
\end{align*}
where in the last line we used the fact that $\bra{\perp_0^\pm(s_{j+1})}(\ket{0}\otimes I)=0$. The fact that no transitions out of  $\{ \ket{\eta_0(s)}, \ket{\eta_1}\}$ can occur during the randomized evolutions means that 
if we are promised that we initialize the system at $s=0$ in the state $\ket{\eta_0(0)} = \ket{0} \ket{E_0^+(0)}$, at every subsequent randomization, then the action of the previous channel has an identical action to the channel
\begin{align}
    \rho \mapsto P(s_j) \rho P(s_j) + \mathcal{C}_{s_j} \circ (I- P(s_j)) \rho (I- P(s_j)),
\end{align}
where $P(s_j) = \ketbra{0}{0} \otimes P_{H(s)}$, with $P_{H(s)}$ the projector onto the zero eigenspace of $H(s)$ and $\mathcal{C}_{s_j}$ is a channel that preserves orthogonality as stated before.

\subsection{Improved Poisson protocol analysis}

Following Ref.~\cite{cunningham2024eigenpath}, we set up a stochastic differential equation for the combined randomization over $s$ (locations of the randomization) and $m$ (number of applications of the walk operator at a given location, sampled according to \eqref{eq:pjintegeres}), where the latter is described by a Poisson process $\mathcal{N}$. Then 
\begin{align}
    d \hat{\rho} = ( W^m(s) \hat{\rho} W^{\dag \, m}(s) - \hat{\rho}) d \mathcal{N}.
\end{align}
Denote by $\mathbb{E}_{\mathcal{N},m}$ the average over both $\mathcal{N}$ and the random variable associated to~$m$, and $\mathbb{E}_m$ the average over the random variable associated to~$m$. The Poisson process is independent of the randomization over $m$, so 
\begin{align}
\nonumber
    d\rho := \mathbb{E}_{\mathcal{N},m}[    d \hat{\rho}] 
    &= \mathbb{E}_m[ W^m(s) \hat{\rho} W^{\dag \, m}(s) - \hat{\rho}]\lambda(s) ds, \\
    & = (P(s) \rho P(s) \nonumber \\ &  + \mathcal{C}_{s} \circ (I- P(s)) \rho (I- P(s)) - \rho) \lambda(s) ds,
    \nonumber
\end{align}
where we recall that $\lambda(s)$ is the rate of the Poisson process. So, we obtain the differential equation
\begin{align}
\label{eq:Evolutionofrho}
    \frac{d\rho}{ds} = (P(s) \rho P(s) + \mathcal{C}_{s} \circ  ( P(s)^\perp \rho P(s)^\perp) - \rho) \lambda(s),
\end{align}
where $P(s)^\perp = I- P(s)$. We set
$$\rho(0)= \ketbra{0}{0} \otimes \ketbra{E^+_0}{E^+_0}.$$
Recall that the target output state is $\ket{\eta_0(1)} = \ket{0} \ket{E^+_0(1)}$. The infidelity error $\gamma$ is given by
\begin{align}
    \gamma & = 1- \mathrm{Tr}[ \rho(1) P(1)] = \mathrm{Tr}[ \rho(0) P(0)] - \mathrm{Tr}[  \rho(1) P(1)] \nonumber \\ & = - \int_0^1 \frac{d}{ds}\left( \mathrm{Tr} [ \rho(s) P(s)]\right) ds.
    \label{eq:gammabound}
\end{align}
 Here, we used the fact that no transition to $\ket{E^-_0}$ happens during randomized evolutions, so we can replace $\ket{E_0^+(s)}\!\bra{E_0^+(s)}$ with $P(s)$.

Let us compute
\begin{align}
\nonumber
    \frac{d}{ds}\mathrm{Tr} [ \rho(s) P(s)] & =   \mathrm{Tr} \left [ \frac{d\rho(s)}{ds} P(s) \right ] +  \mathrm{Tr} \left [ \rho(s) \frac{dP(s)}{ds} \right ]
\\ \nonumber \! & \! \stackrel{\eqref{eq:Evolutionofrho}}{=} \lambda(s)\mathrm{Tr} [ P(s) \mathcal{C}_s \circ (P(s)^\perp \rho P(s)^\perp)] \nonumber \\ & +  \mathrm{Tr} \left[ \rho(s) \frac{dP(s)}{ds}\right] \nonumber
\\ & = \mathrm{Tr} \left[ \rho(s) \frac{dP(s)}{ds}\right],
\label{eq:Poissonproofintermediate1}
\end{align}
where in the last line we used that $$\mathrm{Tr} [ P(s) \mathcal{C}_s\circ (P(s)^\perp \rho P(s)^\perp)] = 0,$$ since $\mathcal{C}_s$ preserves orthogonality.

In Eq.~\eqref{eq:Poissonproofintermediate1}, $\rho(s)$ is not known, but we can solve Eq.~\eqref{eq:Evolutionofrho} for $\rho(s)$
\begin{align}
    \rho = - \frac{d\rho/ds}{\lambda} + P \rho P + \mathcal{C}\circ(P^\perp \rho P^\perp),
\end{align}
where we dropped the $s$ dependence to avoid notational clutter. Using this in Eq.~\eqref{eq:Poissonproofintermediate1}:
\begin{align}
\nonumber
     \frac{d}{ds}\mathrm{Tr} [ \rho P] &= - \mathrm{Tr}\left[\frac{dP}{ds}\frac{d\rho/ds}{\lambda}\right] + \mathrm{Tr}\left[\frac{dP}{ds} P \rho P\right] \nonumber \\ & + \mathrm{Tr}\left[\frac{dP}{ds} \mathcal{C} (P^\perp \rho P^\perp)\right] \nonumber
     \\ & = -\mathrm{Tr}\left[\frac{dP}{ds}\frac{d\rho/ds}{\lambda}\right].
     \label{eq:Poissonproofintermediate2}
\end{align}
The second and third term in the previous equation can be readily seen to vanish. For the second term:
\begin{align}
\nonumber
    \frac{dP}{ds} &= \frac{d (P P)}{ds} =     \frac{dP}{ds} P + P     \frac{dP}{ds}  \\ \nonumber
    & \Rightarrow P     \frac{dP}{ds} P = 0 \\ & \Rightarrow \mathrm{Tr}\left[ \frac{d P}{ds} P\rho P\right]= 0. 
\end{align}
And for the third term:
\small
\begin{align}
\nonumber
\mathrm{Tr}\left[\frac{dP}{ds} \mathcal{C} (P^\perp \rho P^\perp)\right] & =  \mathrm{Tr}\left[\left(\frac{dP}{ds} P + P \frac{dP}{ds} \right)\mathcal{C}  (P^\perp \rho P^\perp)\right]   \\
& =  \mathrm{Tr}\left[\frac{dP}{ds} P \mathcal{C} (P^\perp \rho P^\perp)\right]  \nonumber \\ & + \mathrm{Tr}\left[\frac{dP}{ds}  \mathcal{C} (P^\perp \rho P^\perp) P \right]  =0,
\end{align}
\normalsize
again since $\mathcal{C}$ preserves orthogonality.

Substituting Eq.~\eqref{eq:Poissonproofintermediate2} in Eq.~\eqref{eq:gammabound},
\small
\begin{align}
\nonumber
    \gamma & = \int_0^1 \mathrm{Tr}\left[\frac{dP}{ds}\frac{d\rho/ds}{\lambda}\right] ds \\ & =  \int_0^1 \frac{d}{ds}\left\{ \frac{1}{\lambda} \mathrm{Tr}\left[\frac{dP}{ds}\rho \right] \right\}ds - \int_{0}^1 ds \frac{d}{ds} \left( \frac{1}{\lambda} \right) \mathrm{Tr} \left[ \frac{dP}{ds}\rho\right] \nonumber \\
    & - \int_{0}^1 ds \frac{1}{\lambda} \mathrm{Tr} \left[ \frac{d^2P}{ds^2}\rho\right] \nonumber \\
   &  = \frac{1}{\lambda(1)} \mathrm{Tr}\left[\frac{dP}{ds}(1)\rho(1) \right] - \frac{1}{\lambda(0)} \mathrm{Tr}\left[\frac{dP}{ds}(0)\rho(0) \right] \nonumber \\ & - \int_{0}^1 ds \frac{d}{ds} \left( \frac{1}{\lambda} \right) \mathrm{Tr} \left[ \frac{dP}{ds}\rho\right] 
 - \int_{0}^1 ds \frac{1}{\lambda} \mathrm{Tr} \left[ \frac{d^2P}{ds^2}\rho\right] 
   \nonumber \\
   &  \leq \frac{1}{\lambda(1)} \left\|\frac{dP}{ds}(1) \right\| - \frac{1}{\lambda(0)} \mathrm{Tr}\left[\frac{dP}{ds}(0)\rho(0) \right]  \nonumber \\ & + \int_{0}^1 ds \left|\frac{d}{ds} \left( \frac{1}{\lambda} \right)\right|  \left\| \frac{dP}{ds}\right\| + \int_{0}^1 ds \frac{1}{\lambda} \left\| \frac{d^2P}{ds^2}\right\|. 
   \label{eq:gammaboundintermediate}
\end{align}
\normalsize
So far, we have essentially followed the derivations of Ref.~\cite{cunningham2024eigenpath} in bounding the error $\gamma$, albeit applied to a different setup involving random applications of the walk operators. The rest of the derivation improves over the previous results to return better constant factors.

We now simplify Eq.~\eqref{eq:gammaboundintermediate} by resolving its terms. Let's start from
\begin{align}
\label{eq:P}
    P(s) &= \ketbra{0}{0} \otimes P_{H(s)} \nonumber \\ & = \ketbra{0}{0} \otimes ( \ketbra{E_0^+(s)}{E^+_0(s)} + \ketbra{E^-_0}{E^-_0}). 
\end{align}
Hence, since $\ket{E^-_0}$ has no $s$-dependence,
\small
\begin{align}
\label{eq:Pderivative}
    \frac{dP(s)}{ds} = \ketbra{0}{0} \otimes \left(\frac{d}{ds}\ket{E_0^+(s)}\bra{E^+_0(s)}+\ket{E_0^+(s)}\frac{d}{ds}\bra{E^+_0(s)} \right). 
\end{align}
\normalsize
As shown in the proof of Lemma~\ref{lem:variancebound} below, by exploiting a phase freedom in the definition of $\ket{E^+_0(s)}$, we can make sure that  $\bra{E^+_0(s)} \frac{d}{ds} \ket{E^+_0(s)} =0$. Since $\rho(0)= \ketbra{0}{0} \otimes \ketbra{E^+_0}{E^+_0}$, it follows that the second term in Eq.~\eqref{eq:gammaboundintermediate} is zero:
\begin{align}
\nonumber
    \gamma & \leq \frac{1}{\lambda(1)} \left\|\frac{dP}{ds}(1) \right\| + \int_{0}^1 ds \left|\frac{d}{ds} \left( \frac{1}{\lambda} \right)\right|  \left\| \frac{dP}{ds}\right\| \nonumber \\ & + \int_{0}^1 ds \frac{1}{\lambda} \left\| \frac{d^2P}{ds^2}\right\|. 
\end{align}
Let's introduce two lemmas, bounding $\left\| \frac{dP}{ds}\right\|$ and $\left\| \frac{d^2P}{ds^2}\right\|$ in the expression above:
\begin{lemma}
\label{lem:variancebound}
    \begin{align}
    \label{eq:variancebound}
        \left\| \frac{dP}{ds} \right\| \leq \frac{c_1}{\Delta(s)}, 
    \end{align}
    where 
    \begin{align}
        c_1 = \sqrt{\mathrm{Var}_{\ket{E_0^{+}(s) }}\left[ \frac{dH(s)}{ds}\right]}
    \end{align}
    and for an operator $M$ we defined $$\mathrm{Var}_{\ket{\psi}}[M] = \langle M^2 \rangle - \langle M \rangle^2,$$ with the average taken with respect to $\ket{\psi}$. We find
    \begin{align}
        c_1 \leq \sqrt{2}.
    \end{align}
\end{lemma}
\begin{proof}
From Eq.~\eqref{eq:P} we have
   \begin{align}
       \left\| \frac{d P(s)}{ds} \right\| =    \left\| \frac{d P_{H(s)}}{ds} \right\|.
   \end{align}
Reasoning as in Eq.~\eqref{eq:Pderivative},
\begin{align}
\label{eq:derivativeofP}
    \frac{d P_{H(s)}}{ds} = \frac{d}{ds}(\ket{E^+_0(s)}) \bra{E^+_0(s)} + \ket{E^+_0(s)} \frac{d}{ds}(\bra{E^+_0(s)}). 
\end{align}
Note that $\frac{d P_{H(s)}}{ds}$ has no support on $\ket{E_0^-}$, so we can ignore this eigenstate. Correspondingly, for simplicity we use  the notation  $\ket{E_0^+(s)} \equiv \ket{E_0(s)}$. In this way, all relevant eigenvalues/eigenstates of $H(s)$  are labeled by $E_i$, $\ket{E_i}$. Since
\begin{align}
    H(s) \ket{E_i(s)} = E_i(s) \ket{E_i(s)}, 
\end{align}
we have
\begin{align}
    \frac{dH(s)}{ds} \ket{E_i(s)} + H(s) \frac{d}{ds} \ket{E_i(s)} \nonumber  \\ 
    = \frac{dE_i(s)}{ds} \ket{E_i(s)} + E_i(s) \frac{d}{ds} \ket{E_i(s)}.
\end{align}
So, for $i \neq 0$,
\begin{align}
\label{eq:E0Eielement}
    \bra{E_0(s)} \frac{d}{ds} \ket{E_i(s)} = \frac{\bra{E_0(s)} \frac{dH(s)}{ds} \ket{E_i(s)}}{E_i - E_0} := d_{0i}.
\end{align}
Also, $d_{i0} = - d^*_{0i}$.

For $i=0$ instead, without loss of generality we can choose
\begin{align}
\label{eq:phasechoice}
    \bra{E_0(s)} \frac{d}{ds} \ket{E_0(s)} =0, \quad \forall s.
\end{align}
Let's see that this is the case. $\ket{E_0(s)}$ is defined modulo a phase, $\ket{\bar{E}_0(s)} = e^{i \theta(s)} \ket{E_0(s)}$. Therefore,
\begin{align}
    \frac{d}{ds} \ket{\bar{E}_0(s)} = e^{i \theta(s)} i \frac{d\theta(s)}{ds} \ket{E_0(s)} + e^{i \theta(s)} \frac{d}{ds} \ket{E_0(s)} \; .
\end{align}
Hence,
\begin{align}
     \bra{\bar{E}_0(s)} \frac{d}{ds} \ket{\bar{E}_0(s)} & = e^{- i \theta(s)} \bra{E_0(s)} \frac{d}{ds} \ket{\bar{E}_0(s)} \nonumber \\ & = i \frac{d \theta(s)}{ds} + \bra{E_0(s)} \frac{d}{ds} \ket{E_0(s)}. 
\end{align}
So we can set 
\begin{align}
    \frac{d \theta(s)}{ds} = \bra{E_0(s)} i \frac{d}{ds} \ket{E_0(s)}, \quad \theta(0) = 0. 
\end{align}
We can then choose a phase such that Eq.~\eqref{eq:phasechoice} holds. And now we will drop the bar for simplicity.

Using Eq.~\eqref{eq:E0Eielement} and Eq.~\eqref{eq:phasechoice}:
    \begin{align}
    \frac{d P_{H(s)}}{ds} & = \sum_{i \neq 0} \left(\bra{E_i(s)} \frac{d}{ds}(\ket{E_0(s)}) \ketbra{E_i(s)}{E_0(s)} \right. \nonumber \\ & \left.+ \frac{d}{ds}(\bra{E_0(s)})\ket{E_i(s)}\ket{E_0(s)}\bra{E_i(s)} \right)
    \nonumber \\ & = 
    \sum_{i \neq 0} \left(- d_{0i} \ketbra{E_i(s)}{E_0(s)} +d_{0i} \ketbra{E_0(s)}{E_i(s)} \right).
    \nonumber
\end{align}
The singular values of the matrix $    \frac{d P_{H(s)}}{ds}$ are zero and
\small
\begin{align}
    &\sqrt{\sum_{i \neq 0} d_{0i}^2} \leq \frac{1}{\Delta(s)} \sqrt{\sum_{i \neq 0} \left| \bra{E_i(s)} \frac{dH(s)}{ds} \ket{E_0(s)} \right|^2 } 
  \nonumber  \\  
  &= \frac{1}{\Delta(s)} \sqrt{ \bra{E_0(s)}\frac{dH(s)}{ds} (I - \ketbra{E_0(s)}{E_0(s)}) \frac{dH(s)}{ds} \ket{E_0(s)}   }
   \nonumber \\  
   &= \frac{1}{\Delta(s)}\sqrt{\mathrm{Var}_{\ket{E_0^{+}(s) }}\left[\frac{dH(s)}{ds}\right]},
   \nonumber
\end{align}
\normalsize
since $\ket{E_0(s)} = \ket{E_0^+(s)}$ with the previous notation. The result in Eq.~\eqref{eq:variancebound} follows. For the second part
\begin{align}
 \left\| \frac{dP}{ds} \right\| &\leq \frac{1}{\Delta(s)} \sqrt{|\bra{E_0^+(s)} \left(\frac{d H(s)}{ds} \right)^2 \ket{E_0^+(s)}  |} \nonumber \\ 
 & \leq \frac{1}{\Delta(s)}
 \sqrt{|\bra{y(s)} \left(H(1)- H(0) \right)^2 \ket{y(s)}  |} \nonumber
 \\ & \leq \frac{1}{\Delta(s)}\| H(1) - H(0) \| \nonumber
 \\ & \leq \frac{\sqrt{2}}{\Delta(s)}. 
 \nonumber
\end{align}
To prove the last line, note that
\begin{align}
    \| H(1) - H(0)\| = \| \ketbra{0}{1} \otimes B + \ketbra{1}{0} \otimes B^\dag \|  
\end{align}
where $B = (Z \otimes I + X \otimes \bar{A}) \Pi$. Then,
\begin{align*}
    (\ketbra{0}{1} \otimes B + \ketbra{1}{0} \otimes B^\dag)^\dag (\ketbra{0}{1} \otimes B + \ketbra{1}{0} \otimes B^\dag) \nonumber \\ = \ketbra{0}{0} \otimes B B^\dag + \ketbra{1}{1} \otimes B^\dag B,
\end{align*}
which means that 
\begin{align}
     \| H(1) - H(0)\| = \| B \| \leq \| - Z \otimes I + X \otimes \bar{A} \|
\end{align}
Now,
\begin{align}
\nonumber
    ( - Z \otimes I + X \otimes \bar{A})^\dag ( - Z \otimes I + X \otimes \bar{A}) = I \otimes (I + \bar{A}^2)
\end{align}
is a matrix whose eigenvalues are bounded by $2$, because $\|\bar{A}^2\| \leq \| \bar{A}\|^2 \leq 1$. Hence $\|  - Z \otimes I + X \otimes \bar{A}\| \leq \sqrt{2}$ and so $ \| H(1) - H(0)\|  \leq \sqrt{2}$ immediately follows.
\end{proof}
This lemma strengthens the bound in  Lemma~3 of Ref.~\cite{cunningham2024eigenpath}, which gave $\left\| \frac{dP}{ds}\right\| \leq \frac{2 \sqrt{2}}{\Delta(s)}$. 

\begin{lemma}
\label{lem:seconderivariveP}
Assume $\frac{d^2 H}{ds^2} = 0$. Then
    \begin{align}
        \left\| \frac{d^2P}{ds^2} \right\| \leq  6  \frac{\| \frac{dH}{ds} \|^2}{\Delta^2(s)}. 
    \end{align}
\end{lemma}
For us, we have that $\| \frac{dH}{ds} \| \leq \sqrt{2}$, and so  
     \begin{align}
        \left\| \frac{d^2P}{ds^2} \right\| \leq \frac{12}{\Delta(s)^2}. 
    \end{align}
\begin{proof}
    Recall that from Eq.~\eqref{eq:P} $P(s) = \ketbra{0}{0} \otimes P_{H(s)}$
    First write 
    \begin{align}
        P_{H(s)} &= -\frac{1}{2\pi i} \oint_{\Gamma} R_z dz
    \end{align} 
    where $\Gamma$ is a circle of radius $\Delta(s)/2$ in the complex plane with center $0$, and the contour integral is taken in the anti-clockwise direction. We also have 
    \begin{align}
        R_z = (z I - H(s))^{-1}
    \end{align}
    being the resolvent. Using $\frac{d M^{-1}}{ds} = -M^{-1} \frac{dM}{ds} M^{-1}$, we have $\frac{d R_z}{ds} =  R_z \frac{dH}{ds} R_z$. With this
\small
\begin{align}
    \frac{d^2 P_{H(s)}}{ds^2} &= - \frac{1}{2\pi i} \oint_{\Gamma} \frac{d^2 R_z}{ds^2} dz \nonumber \\ & =  -\frac{1}{2\pi i} \oint_{\Gamma}  \left( 2 R_z  \frac{dH}{ds} R_z \frac{dH}{ds} R_z \nonumber + R_z  \frac{d^2H}{ds^2} R_z \right)dz \\ \nonumber
    &=  -\frac{1}{2\pi i} \oint_{\Gamma} \left( 2 R_z  \frac{dH}{ds} R_z  \frac{dH}{ds} R_z\right) dz  \\
    & = -2  G\left(   \frac{dH}{ds},  \frac{dH}{ds} \right),
\end{align}
\normalsize
where we used that in our setting $d^2 H(s)/ds^2 =0$ and we introduced
\begin{align}
  G( X,Y) =  \frac{1}{2\pi i} \oint_{\Gamma} \left(  R_z X R_z Y R_z\right) dz.
\end{align}
The operator $G$ was analyzed in Ref.~\cite{jansen2007bounds}. In particular, in the blocks defined by the images of $P_{H(s)}$ and $P_{H(s)}^\perp$, $G$ has the following structure (proof of Lemma~5 \cite{jansen2007bounds}):
\begin{align}
    &G\left(\frac{dH}{ds} ,
\frac{dH}{ds} \right) \nonumber \\ 
    = & P_{H(s)} B_1 P_{H(s)} + P_{H(s)} B_2 P_{H(s)}^\perp  \nonumber \\ 
    &- P_{H(s)}^\perp B_2 P_{H(s)} - P_{H(s)}^\perp B_1 P_{H(s)}^\perp,
\end{align}
 where $B_1= \widetilde{K}^2$, $B_2 = \widetilde{K \widetilde{K }} -\widetilde{\widetilde{K }  K }$, $K =  \frac{dH}{ds}$
and
\begin{align}
    \tilde{X} := \frac{1}{2\pi i } \oint_{\Gamma} R_z X R_z dz,
\end{align}
was called in Ref.~\cite{jansen2007bounds}
the `twiddle operation'. $dH/ds$ is Hermitian, and the twiddle operation preserves Hermiticity. 
Hence $B_1$ is Hermitian and $B_2$ is anti-Hermitian. $G$~is then Hermitian, so its norm just equals the largest absolute value of the eigenvalues. Compute
\begin{align}
\nonumber
    G v &= G P_{H(s)} v + G P_{H(s)}^\perp v \\
    \nonumber
    & = [P_{H(s)} B_1  - P_{H(s)}^\perp B_2] v_1  + [P_{H(s)} B_2  - P_{H(s)}^\perp B_1 ]v_2 \\
    & = P_{H(s)} (B_1 v_1 + B_2 v_2) - P_{H(s)}^\perp (B_2 v_1 + B_1 v_2)
\end{align}
where $v_1 = P_{H(s)} v$, $v_2 = P_{H(s)}^\perp v$ and $v$ is any unit vector. This is the sum of two orthogonal vectors, so the norm of $Gv$ is
\small
\begin{align}
     \| G v \|^2 & = \| P_{H(s)} (B_1 v_1 + B_2 v_2)\|^2 + \|P_{H(s)}^\perp (B_2 v_1 + B_1 v_2)\|^2 \nonumber \\
     & \leq \| (B_1 v_1 + B_2 v_2)\|^2 + \|(B_2 v_1 + B_1 v_2)\|^2
     \nonumber\\
    & = v_1^\dag B_1^2 v_1 + v_2^\dag B_1^2 v_2 - v_1^\dag B_2^2 v_1 - v_2^\dag B_2^2 v_2 \nonumber \\
    & + v_2^\dag (-B_2 B_1 +B_1 B_2) v_1 + v_1^\dag (B_1 B_2 -B_2 B_1)v_2.
\end{align}
\normalsize
If $b_1 = \|B_1\|$, $b_2 = \|B_2\|$, $x_1 = \|v_1\|$, $x_2 = \|v_2\|$, we get
\begin{align}
  \| G v\|^2 & \leq b_1^2 (x_1^2 + x_2^2) + b_2^2 (x_1^2 + x_2^2) + 4 b_1 b_2 x_1 x_2
    \end{align}
    The norm of the matrix is upper bounded by maximizing the last expression over all $x_1 \geq0 $, $x_2 \geq 0$ with \mbox{$x_1^2 + x_2^2 =1$}. Since $x_1 x_2 \leq 1/2$ on the unit circle
    \begin{align}
  \| G v \|^2 & \leq (b_1+b_2)^2.
    \end{align}
    Hence 
\begin{align}
    \|G \| & \leq \|B_1\| + \|B_2\| \nonumber \\
    & = \left\|\widetilde{K}\right\|^2 +  \left\| - \widetilde{K \widetilde{K }} + \widetilde{\widetilde{K}  K}\right\|.
\end{align} 
Using Lemma~7 in Ref.~\cite{jansen2007bounds},
\begin{align}
    \|\tilde{K}\| \leq \frac{\|K\|}{\Delta} = \frac{\left\| \frac{dH}{ds} \right\|}{\Delta}.
\end{align}

Putting all of this together we find
\small
\begin{align}
   \left\| \frac{d^2P}{ds^2} \right\|  & =  \left\| \frac{d^2P_{H(s)}}{ds^2} \right\| \nonumber \\
   &= 2 \| G \|  \nonumber\\
   & \leq 2 \left( \frac{\left\|\frac{dH}{ds} \right\|^2}{\Delta^2} +  \left\| \widetilde{\frac{dH}{ds} \widetilde{\frac{dH}{ds} }} \right\| + \left\|\widetilde{\widetilde{\frac{dH}{ds} }  \frac{dH}{ds} }\right\| \right)\nonumber\\ 
& \leq 2 \left(\frac{\left\|\frac{dH}{ds} \right\|^2}{\Delta^2} +   \frac{\left\| \frac{dH}{ds} \widetilde{\frac{dH}{ds} } \right\|}{\Delta} +  \frac{\left\|\widetilde{\frac{dH}{ds} }  \frac{dH}{ds} \right\|}{\Delta} \right) \nonumber\\
& \leq 2 \frac{\left\|\frac{dH}{ds} \right\|^2}{\Delta^2} + 4  \frac{\| \frac{dH}{ds} \| \| \widetilde{\frac{dH}{ds} } \|}{\Delta} \nonumber \\
& \leq 2 \frac{\left\|\frac{dH}{ds} \right\|^2}{\Delta^2} + 4  \frac{\| \frac{dH}{ds} \|^2}{\Delta^2} \nonumber \\
& = 6  \frac{\| \frac{dH}{ds} \|^2}{\Delta^2}.
\end{align}
\normalsize

In the proof of Lemma~\ref{lem:variancebound} we have seen $
    \left\| \frac{d H}{ds} \right\|  \leq \sqrt{2}$, so the final result follows immediately.
\end{proof}

Note that the bound above strengthens Lemma 3 in Ref.~\cite{cunningham2024eigenpath}, which gave $\left\| \frac{d^2 P}{ds} \right\| \leq 16/\Delta^2$.

With these estimates the error bound becomes
\small
\begin{align}
\nonumber
    \gamma & \leq \frac{\sqrt{2}}{|\lambda(1)| \Delta(1)} + \int_{0}^1 ds \frac{\sqrt{2}}{\Delta(s)} \left|\frac{d}{ds} \left( \frac{1}{\lambda} \right)\right|  + \int_{0}^1 ds \frac{12}{|\lambda| \Delta^2} \nonumber \\ & = (1) + (2) + (3) . 
\end{align}
\normalsize
Now, in our case $\Delta(s) = \sqrt{(1-s)^2 + (s/\kappa)^2}$. Furthermore, we shall take the rate
\begin{align}
\label{eq:lambdas}
    \lambda(s) = \frac{C}{\Delta(s)^q \Delta^{1-q}_{\textrm{min}}},  
\end{align}
where $\Delta^{1-q}_{\textrm{min}}$ is the minimal gap, which occurs at $s = (1+1/\kappa^2)^{-1}$ and equals
\begin{align}
    \Delta_{\textrm{min}} = \sqrt{\frac{1}{1+\kappa^2}}.
\end{align}
Now,
\begin{align}
  \left |  \frac{d}{ds} \frac{1}{\lambda} \right | = \frac{q \Delta^{q-1} |-1+s + s/\kappa^2|}{C \Delta}\Delta_{\textrm{min}}^{1-q}
\end{align}
and $\Delta(1) = 1/\kappa$, $\lambda(1) = \frac{C \kappa^q}{\Delta^{1-q}_{\textrm{min}}}$. So
\begin{align*}
    (1) &= \frac{\sqrt{2} \Delta^{1-q}_{\textrm{min}}}{C \kappa^{q-1}},
   \\  (2) & = \int_0^1 ds \frac{\sqrt{2} q \Delta(s)^{q-3}}{C} \Delta_{\textrm{min}}^{1-q} \left| -1+ s + \frac{s}{\kappa^2} \right|, \\ 
    (3) & = \int_0^1 ds \frac{12 \Delta(s)^{q-2} \Delta_{\textrm{min}}^{1-q}}{C}.
\end{align*}
 We find that taking $q=1/2$ gives the minimal error. We next establish a sufficient value for $C$ to ensure that $\gamma \le 1/2$.

For $q=1/2$ we have that
\begin{align*}
    (1) &= \frac{\sqrt{2} \Delta^{1/2}_{\textrm{min}} \kappa^{1/2}}{C},
   \\  (2) & = \frac{\sqrt{2}\Delta_{\textrm{min}}^{1/2} }{2C} \int_0^1 ds   \frac{\left| -1+ s + \frac{s}{\kappa^2} \right|}{\Delta(s)^{5/2}}, \\ 
    (3) & = \frac{12\Delta_{\textrm{min}}^{1/2} }{C}\int_0^1 ds \frac{1}{\Delta(s)^{3/2} }.
\end{align*}
We first note that $\Delta_{\textrm{min}} = 1/\sqrt{1+\kappa^2}$, and for the first term we have
\begin{equation}
    (1) = \frac{2}{C} \sqrt{\frac{\kappa}{1+\kappa^2}} \le \frac{\sqrt{2}}{C}.
\end{equation}
The second term can be written as
\begin{equation}
    (2) = \frac{\sqrt{2}\Delta_{\textrm{min}}^{1/2} }{2C} \int_0^1 ds   \frac{\left|\frac{d\Delta}{ds} \right|}{\Delta(s)^{3/2}},
\end{equation}
however $\frac{d\Delta}{ds} \ge 0$ for $s\in [s_{\min}, 1]$, and $\frac{d\Delta}{ds} < 0$
 for $s \in [0,s_{\min})$, where $s_{\min} = \kappa^2/(1+\kappa^2)$. This implies that
 \small
\begin{align}
     (2) &= \frac{\sqrt{2}\Delta_{\textrm{min}}^{1/2} }{2C}  \left ( -\int_0^{s_{\min}} ds   \frac{\frac{d\Delta}{ds} }{\Delta(s)^{3/2}} +\int_{s_{\min}}^1 ds   \frac{\frac{d\Delta}{ds}}{\Delta(s)^{3/2}} \right ) \nonumber \\
     &= \frac{\sqrt{2}\Delta_{\textrm{min}}^{1/2} }{2C}  \left ( 2\int_0^{s_{\min}} ds   \frac{d (\Delta^{-1/2})}{ds}  -2\int_{s_{\min}}^1 ds   \frac{d(\Delta^{-1/2})}{ds} \right ) \nonumber \\
     &= \frac{1 }{\sqrt{2}C} \frac{1}{(1+\kappa^2)^{1/4}}  \left (4 (1+\kappa^2)^{1/4} - 2( 1 + \sqrt{\kappa})\right ) \nonumber \\
     & \le \frac{\sqrt{2}}{C}.
\end{align}
 \normalsize
 Finally, the third term is given by
 \begin{equation}
     (3) =\frac{12\Delta_{\textrm{min}}^{1/2} }{C}\int_0^1 ds \frac{1}{((1-s)^2 + (s/\kappa)^2)^{3/4} }.
 \end{equation}
We want to upper bound the integral
\begin{align}
    I_\kappa &= \int_0^1 ds \frac{1}{((1-s)^2 + (s/\kappa)^2)^{3/4} } \, .
\end{align}
First we perform a change of variables 
 \begin{equation}
     \tilde{s} := \frac{1+\kappa^2}{\kappa} ( s - s_{\min}).
 \end{equation}
to obtain
\begin{align}
    \label{eq:t1}
    I_\kappa &= \frac{\kappa}{(\kappa^2 +1)^{1/4}} \int_{-\kappa}^{1/\kappa}d\tilde{s} \frac{1}{(1+\tilde{s}^2)^{3/4}} \, .
\end{align}
The integral on the RHS can be split into two integrals over $[-\kappa,0]$ and $[0,1/\kappa]$.  Changing variables $v=1/u$ we can rewrite the first integral as
\begin{align}
    \int_{-\kappa}^{0}d\tilde{s} \frac{1}{(1+\tilde{s}^2)^{3/4}} &= \int_{1/\kappa}^\infty dv v^{-1/2} \frac{1}{(1+v^2)^{3/4}} \\
    & \le \int_{1/\kappa}^\infty dv  \frac{1}{(1+v^2)^{3/4}} \, ,
\end{align}
where for the inequality we used the fact that $v^{-1/2}<1$ in the domain of integration. 
Combining the integrals over both domains we obtain
\begin{align}
    \int_{-\kappa}^{1/\kappa}d\tilde{s} \frac{1}{(1+\tilde{s}^2)^{3/4}} &\le \int_0^\infty d\tilde{s} \frac{1}{(1 + \tilde{s}^2)^{3/4}}\, , \\
    &= \frac{2 \sqrt{\pi} \Gamma(5/4)}{\Gamma(3/4)}\, ,
\end{align}
where $\Gamma(x)$ is the Gamma function. Putting this back into Eq.~\eqref{eq:t1} we obtain
\begin{align}
    \label{eq:t3}
    I_\kappa &\le \left( \frac{2 \sqrt{\pi} \Gamma(5/4)}{\Gamma(3/4)} \right) \frac{\kappa}{(\kappa^2+1)^{1/4}} \, \kappa^{1/2} \, .
\end{align}
which then implies
\begin{align}
          \label{eq:intbound}
        (3)  &  \le \frac{24 \sqrt{\pi} \Gamma(5/4) }{C \Gamma(3/4)}\, .
\end{align}
 Therefore, we have
 \begin{equation}
     \gamma \le (1)+(2) +(3) \le \frac{1}{C} \left (2 \sqrt{2} + \frac{24 \sqrt{\pi} \Gamma(5/4) }{\Gamma(3/4)} \right).
 \end{equation}
 Choosing $C \ge 68.59$ then implies that $\gamma \le 1/2$.
 
 The overall average query cost is then (recall Eq.~\eqref{eq:querycostats} and Eq.~\eqref{eq:lambdas})
 \small
\begin{align}
\int_0^1 \lambda(s) Q(s)  ds & \leq   2 \times 2.322 \alpha \int_0^1 \frac{\lambda(s)}{\Delta(s)}   ds   \nonumber  \\ &  \leq
318.6 \alpha (1+\kappa^2)^{1/4} \int_{0}^1 \frac{1}{[(1-s)^2 + (s/\kappa)^2]^{3/4}},
\nonumber
\end{align}
\normalsize
 where we use that $\alpha \ge \alpha(s)$ for all $s$ (not using this inequality gives little to no improvement). Using Eq.~\eqref{eq:t3}
\begin{align}
\int_0^1 \lambda(s) Q(s)  ds \leq 318.6 \frac{ 2 \sqrt{\pi} \Gamma(\frac{5}{4})}{\Gamma(\frac{3}{4})} \alpha \kappa \approx 835.4 \alpha \kappa
\end{align}
\normalsize

\subsection{Filtering cost}
	\label{app:filteringcost}
	
	We next consider applying $P(1)$ to the output state of the adiabatic component of the protocol, where $P(1)$ is the projector onto the nullspace of the Hamiltonian $H(1)$. A QSP implementation was given in Ref.~\cite{lin2020optimal}, but it requires phase factor precomputations.  
    
    We expand and formalize the analysis of Ref.~\cite{costa2022optimal}, proving stronger guarantees for their algorithm.

    \begin{lemma}[Filtering cost]
		\label{res:filteringcost}
		Let $U_{H}$ be an $(\alpha_H, m_H, 0)$ block encoding of a Hermitian operator $H$,  let $P$ denote the projection onto the nullspace of $H$ and $\Delta$ a lower bound on the gap between zero and the nearest non-zero eigenvalue. Then we can realize a $(1,m_H+2,\epsilon_P)$ block-encoding of $P$ 	with a number of calls to $U_{H}$, $U^\dagger_{H}$ equal to
		\begin{equation}
			l_P= \left \lceil \frac{\alpha_H}{\Delta} \ln \frac{2}{ \epsilon_P} +2 \right \rceil,
		\end{equation}
        with no QSP phase factor precomputation required.
	\end{lemma}

\begin{proof}
We consider a circuit PREP that prepares a  `window' state $\ket{\psi}$, given as
\begin{align}
    \ket{0} \mapsto \ket{\psi}:=\sum^l_{j=0} \sqrt{\psi_j} \ket{j}. 
\end{align}
The SELECT is
\begin{align}
    \sum_{j=0}^l \ketbra{j}{j} \otimes W_{2}^j,
\end{align}
where $W_2$ is the walk operator 
\begin{align}
    W_2 = U^\dag_H \mathcal{Z} U_H \mathcal{Z},
\end{align}
and $\mathcal{Z} = (2 \ketbra{0}{0} - I)\otimes I$ where $\ket{0}$ labels the block-encoding space. The overall protocol then involves the standard $\mathrm{PREP}^\dag \circ \mathrm{SELECT} \circ \mathrm{PREP}$ routine. The Hermitian operator $H$ is assumed to have eigendecomposition
\begin{equation}
    H = \sum_{m,s} \lambda_{m} \ketbra{\varphi_{m,s}}{\varphi_{m,s}},
\end{equation}
where $\{\ket{\varphi_m}\}$ is a normalized basis of eigenvectors of $H$ and $s$ is a degeneracy index. We label the nullspace by $m=0$, so that $\lambda_{0,s} = 0$.

The measurement operator induced on the system by finding the ancillas where the window state was prepared in the zero state is
\begin{align}
    R' &= (\bra{\psi} \otimes I) \mathrm{SELECT} (\ket{\psi} \otimes I) \\
     & = \sum_{j=0}^l \psi_j W_2^j. 
\end{align}
The walk operator can be written as (see Section~\ref{sec:randomizedwalkmethod}) 
\begin{align}
    W_2 = \sum_{m,s} e^{i 2 \phi_m} \ketbra{\xi^+_{m,s}}{\xi^+_{m,s}} + e^{-2 i \phi_{m}} \ketbra{\xi^-_{m,s}}{\xi^-_{m,s}},
\end{align} 
where 
\begin{align}
    \ket{\xi_{m,s}^{\pm}} & = \frac{1}{\sqrt{2}} (\ket{0} \ket{\varphi_{m,s}} \mp i \ket{\perp_{m,s}}), \\
   \phi_m & = \arccos(\lambda_m/\alpha),
\end{align}

 Let us define $\tilde{\psi}(\zeta) := \sum_{j=0}^l \psi_j e^{i \zeta j}$.
Then
\begin{align}
    R' & = \sum_m \tilde{\psi}(2\phi_m) \Pi^+_m + \tilde{\psi}(-2\phi_m)  \Pi^-_m,
\end{align}

with $  \Pi^\pm_m = \sum_s \ketbra{\xi^\pm_{m,s}}{\xi^\pm_{m,s}}$.
We also have that
\begin{align}
    R &= (\bra{0} \otimes I) R' (\ket{0} \otimes I) 
    \\ & = 
  \sum_m \frac{1}{2} (\tilde{\psi}(2\phi_m) + \tilde{\psi}(-2\phi_m) ) P_m,
\end{align}
where $P_m = \sum_s \ketbra{\varphi_{m,s}}{\varphi_{m,s}}$. 
For the nullspace $m=0$
we have
\begin{equation}
    2\phi_0 = 2\arccos(0) = \pi \; .
\end{equation}

 If we require $\tilde{\psi}(\zeta = \pm \pi) = 1$, it follows that
\begin{align}
    \| R - P\| & \leq \frac{1}{2} \max_{m \not= 0}  |  \tilde{\psi}( 2\phi_m) + \tilde{\psi}(-2\phi_m)| \; .
\end{align}
Following the same reasoning seen in Sec.~\ref{sec:randomizedwalkmethod}, $\phi_m$ takes values in $[0,\pi]$, with a gap around $\phi_0=\pi/2$ of at least $\Delta /\alpha$. 
So
\begin{align}
    \| R - P\| & \leq \frac{1}{2} \max_{\zeta \not\in [\pi - 2 \Delta/\alpha_H, \pi+ 2 \Delta/\alpha_H]}  |  \tilde{\psi}(\zeta) + \tilde{\psi}(-\zeta)|.
\end{align}

Following~\cite{costa2022optimal}, we choose $\psi_j$ to provide a Dolph-Chebyshev window function, with Fourier transform
\begin{equation}
    \tilde{\psi}(2\theta_m) = \epsilon_P T_l [ \beta \cos  ( 2\theta_m-\pi)],
\end{equation}
where  $\theta_m = \frac{\pi}{2}(1+ \frac{m}{l})$, with $m=-l, \dots, l$  and $T_l(x)$ are Chebyshev polynomials of the first kind.
Here $\beta$ is the trade-off parameter between main lobe and side lobe weights. We fix it as
\begin{equation}
    \beta = \cosh \left ( \frac{1}{l} \cosh^{-1} (\frac{1}{\epsilon_P}) \right ),
\end{equation}
where  $\epsilon_P$ is the required attenuation.  Then, with this choice we have, as required above,
\begin{align}
      \tilde{\psi}(\zeta = \pm \pi)=  \epsilon_P T_l [\beta] =1.
\end{align}
Furthermore, we want to set the width of the peak to coincide with the phase gap, which equals $2\Delta/\alpha_H$.
 
Since $T_l(1)=1$, we set the relation between $\beta $ and $\Delta $ to be
\begin{equation}
    \beta \cos (2\Delta/\alpha_H) = 1.
\end{equation}
Combining the above equations gives
\begin{equation}
    l = \left\lceil \frac{\cosh^{-1} (1/\epsilon_P)}{\cosh^{-1} (1/(\cos(2\Delta/\alpha_H)))} \right\rceil \le \frac{\alpha_H}{2\Delta} \ln \frac{2}{\epsilon_P} +1.
\end{equation}
 With this choice, 
\begin{align}
    \max_{\zeta \not\in [\pi - 2 \Delta/\alpha_H, \pi+ 2 \Delta/\alpha_H]}  |  \tilde{\psi}(\pm \zeta)| \leq \epsilon_P,
\end{align}
and so  $\| R - P\|  \leq \epsilon_P$.

However, each call to the walk operator involves one call to $U_H$ and one call to $U^\dagger_H$ and so $l_P = 2 l\leq  (\alpha_H/\Delta )\ln (2/\epsilon_P) +2 $, as claimed. The circuit implementation of the above is the same as in~\cite{costa2022optimal}, using a just-in-time unary encoding of the controlled walk operators, where the uncompute ($\mathrm{PREP}^\dagger$) of the window function state is inverted in its controls. The orderings of these uncompute control operations are then changed so that at most 2 qubits are ever required coherently in the control register.
\end{proof}
As highlighted in~\cite{costa2022optimal}, the use of the just-in-time method also means that mid-circuit measurements can reduce circuit costs by terminating the circuit earlier than a pure post-selection circuit.

    Here we apply this with $H=H(1)$, $\alpha_H = \alpha$, $\Delta = \Delta(1) = 1/\kappa$. 
    What is more, a factor of $2$ is saved by flagging a failure early.

\subsection{Total qubit count}
We now specify the total number of logical qubits needed for the algorithm. The following auxiliary qubits are required by the algorithm:
\begin{itemize}
    \item A total of $a$ auxiliary qubits for the block-encoding $U_A$ of the matrix $A$.
    \item A single qubit for the Hermitian extension $A \rightarrow \bar{A}$.
    \item A single qubit for $\bar{A} \rightarrow A(s)$.
    \item A single qubit for $A(s) \rightarrow H(s)$.
    \item Two extra qubits for block-encoding $H(s) \rightarrow U_{H(s)}$.
    \item Two extra qubits for performing the eigenspace filtering.
\end{itemize}
This gives a total number of auxiliary qubits being $a+7$. Therefore, with $A$ acting on $n$ qubits, the total number of logical qubits required for the algorithm is
\begin{equation}
    n_L = n+a+7,
\end{equation}
as claimed in the main text.
\subsection{Combining adiabatic and filtering errors}
	\label{app:adiabaticfilteringerrors}
	
	In this section we determine how the adiabatic error $\epsilon_{AD}$ combines with the filtering error $\epsilon_P$.     We will repeatedly use the following simple results:
\begin{lemma}
	\label{lem:1normbound}
	For any two square matrices $A$, $B$,
	\begin{equation}
		\|A B\|_1 \leq \min\{ \|A\|_1 \|B\| ,  \|A\| \|B\|_1\}.
	\end{equation}
\end{lemma}
\begin{proof}
	From the definition of trace-norm \cite{watrous2018theory},
	\begin{equation}
		\|A B\|_1 = \max_{U} |\langle AB, U \rangle|,
	\end{equation}
	where the maximization is over all unitaries and $\langle \cdot, \cdot \rangle$ is the Hilbert-Schmidt scalar product. Using H\"older's inequality for Schatten $p$-norms and their invariance under composition with unitaries we have
	\begin{equation*}
		|\langle AB, U \rangle| = |\langle A, U B^\dag \rangle| \leq \|A\|_1 \| U B^\dag \|  = \|A\|_1 \|B\|.
	\end{equation*}
	Similarly, we have that $|\langle A, U B^\dag \rangle| \leq \|A\|\| U B^\dag \|_1 =\|A\| \|B\|_1$.  Therefore we can take the smaller of the two results as the upper bound on $\|A B\|_1$, which completes the proof.
\end{proof}
    
    \begin{lemma}[Perturbing $1$-norm]
	\label{lem:stabilitynorm}
	Assume $\|\tilde{A} - A\| \leq \delta$, $\|\tilde{B} - B \| \leq \delta$. Then for every density operator $\rho$,
	\begin{equation}
		\|\tilde{A} \rho \tilde{B} - A \rho B \|_1 \leq \delta( \|A\| + \|B\|) + \delta^2.
	\end{equation} 
	\begin{proof}
		Using Lemma~\ref{lem:1normbound},
		\small
		\begin{align}
			&\|\tilde{A} \rho \tilde{B} - A \rho B \|_1 	\nonumber  \\ 
            &= \| (\tilde{A} - A)\rho (\tilde{B} - B) + (\tilde{A} - A) \rho B + A \rho (\tilde{B}- B) \|_1 
			\nonumber \\  
            &\leq \| (\tilde{A} - A)\rho (\tilde{B} - B)\|_1 + \|(\tilde{A} - A) \rho B\|_1 + \|A \rho (\tilde{B}- B) \|_1 \nonumber\\  
            &\leq  \|\tilde{A} - A\| \|\tilde{B} - B\| \|\rho\|_1  + \|\tilde{A} - A\| \| B\| \|\rho\|_1 \nonumber \\ 
            &+  \|\tilde{B} - B\| \| A\| \|\rho\|_1\nonumber \\  &\leq \delta^2 + \delta (\|A\| + \|B\|). \nonumber
		\end{align}
		\normalsize
		Which completes the proof.
	\end{proof}
\end{lemma}

    \begin{lemma}[Perturbing probabilities]
	\label{lem:perturbingprob}
	Let $A$, $\tilde{A}$ be square matrices satisfying $\|A - \tilde{A}\| \leq \delta$ and $\rho$ a density operator. Then 
	\begin{equation}
		\tr{\rho A^\dag A} -2 \delta \|A\| \le \tr{ \rho \tilde{A}^\dag \tilde{A}} \le  \tr{\rho A^\dag A} + 2 \delta \|A\|  +\delta^2,
	\end{equation}
\end{lemma}
\begin{proof}
	Let $\Delta A = \tilde{A} - A$. Then using H\"older's inequality and submultiplicativity
	\small
	\begin{align}
		\tr{ \rho \tilde{A}^\dag \tilde{A}} \nonumber & = \tr{\rho (A^\dag + \Delta A^\dag)(A+\Delta A)} \nonumber \\ & = \tr{\rho A^\dag A} + \tr{\rho (A^\dag \Delta A + \Delta A A^\dag}) \nonumber \\
        &+ \tr{\rho \Delta A^\dag \Delta A}.
	\end{align}	
	\normalsize
	The last term is always positive,  and therefore we have that
	\begin{align}
		\tr{\rho A^\dag A} - | \tr{\rho (A^\dag \Delta A + \Delta A A^\dag})| \le \tr{ \rho \tilde{A}^\dag \tilde{A}}  \nonumber \\  \le  \tr{\rho A^\dag A} + | \tr{\rho (A^\dag \Delta A + \Delta A A^\dag})| +\tr{\rho \Delta A^\dag \Delta A}.
	\end{align}
	Making use of the H\"older inequality $|\tr{A^\dagger B}|  \le \|A\|_1 \|B\|$ and submultiplicativity of the operator norm we have that
	\small
	\begin{align}
		\tr{\rho A^\dag A} -2 \delta \|A\| \le \tr{ \rho \tilde{A}^\dag \tilde{A}} \le  \tr{\rho A^\dag A} + 2 \delta \|A\|  +\delta^2,
	\end{align}
	\normalsize
	as required.

\end{proof}

	\begin{thm}[Adiabatic \& Filtering Errors]
		\label{res:overall}
		Consider the following protocol:
		\begin{enumerate}
			\item Apply the Poisson adiabatic protocol with $\gamma \ge 1/2$, outputting a state $\rho$.
            \item Apply an approximation $R$ of $P_{H(1)}$, where $P_{H(1)}$ is the nullspace projector for $H(1)$, as in Lemma~\ref{res:filteringcost}, and $\|R-P_{H(1)}\| \le \epsilon_P$.
		\end{enumerate}
        Let $\ket{y}$ denote the normalized state $\propto A^{-1} \ket{b}$ and $\tilde{\sigma}$ the state prepared on the same registers by the above protocol. We have 
		\begin{equation}
			\label{eq:errorfinal}
			\left\| \tilde{\sigma} - \ketbra{y}{y} \right\|_1  \leq  8\epsilon_P + 4\epsilon_P^2,
		\end{equation}
		with success probability $p_{succ}$ satisfying
		\begin{equation}
			\label{eq:successprobabilityfinal}
			p_{succ} \geq 	\frac{1}{2} - 2 \epsilon_P.
		\end{equation}
	\end{thm}
	
	\begin{proof}
 As we have shown, the nullspace of $H(1)$ is spanned by $\ket{0,+,1,y}$ and $\ket{1,+,0,b}$, however $\rho$ has zero overlap with $\ket{1,+,0,b}$, due to the form of the adiabatic protocol. Therefore, we have that
\begin{align}
\label{eq:equivalentform}
\frac{1}{\tr{P_{H(1)}\rho}} P_{H(1)} \rho P_{H(1)} =  \ket{0,+,1,y}\bra{0,+,1,y}.
\end{align}
Here, $\tr{P_{H(1)} \rho} \ge 1/2$, from the choice of adiabatic parameters.
        
However, for the second step we apply an approximate projector $R$ to $\rho$ to obtain the quantum state
        \begin{align}
            \frac{1}{\tr{R^\dagger R \rho}} R\rho R^\dagger.
        \end{align}
        The success probability of implementing this is $\tr{R^\dagger R \rho}$, but using Lemma~\ref{lem:perturbingprob} we have that
        \begin{equation}
        \label{eq:probbound1}
            \tr{R^\dagger R \rho} \ge \tr{P_{H(1)} \rho} - 2 \epsilon_P \ge \frac{1}{2} - 2\epsilon_P.
        \end{equation}
        We next bound the $L_1$ norm error in the output. We have from Lemma~\ref{lem:stabilitynorm} that
        \begin{align}
            \|R \rho R^\dagger - P_{H(1)} \rho P_{H(1)}\|_1 \le 2 \epsilon_P + \epsilon_P^2.
        \end{align}

		The distance between the output of the algorithm and the target is then bounded as
		\small \begin{align*}
			&\left\| \frac{R \rho R^\dag}{\tr{ R \rho R^\dag}} -  \frac{P_{H(1)} \rho P_{H(1)}^\dag}{\tr{ P_{H(1)} \rho}}\right\|_1 \leq \nonumber\\ 
				&\left\| \frac{R \rho R^\dag}{\tr{ R \rho R^\dag}} -  \frac{R \rho R^\dag}{\tr{\rho P_{H(1)}}} \right\|_1 + \left\|\frac{R \rho R^\dag}{\tr{\rho P_{H(1)}}} -  \frac{P_{H(1)} \rho P_{H(1)}}{\tr{  \rho P_{H(1)}}}\right\|_1.
		\end{align*}
        \normalsize
		We will bound the two terms separately. First:
		\begin{align}
			\left\| \frac{R \rho R^\dag}{\tr{ R \rho R^\dag}} -  \frac{R \rho R^\dag}{\tr{\rho P_{H(1)}}} \right\|_1 \\ = \|R \rho R^\dag\|_1 \left| \frac{1}{\tr{R \rho R^\dag}} - \frac{1}{\tr{\rho P_{H(1)}}} \right|\nonumber \\ 
			 = \|R \rho R^\dag\|_1 \frac{\left| \tr{R \rho R^\dag} - \tr{ \rho P_{H(1)}} \right|}{ \tr{R \rho R^\dag}   \tr{\rho P_{H(1)}} } 
			\nonumber\\ 
			 = \frac{\left| \tr{R \rho R^\dag} - \tr{ \rho P_{H(1)}} \right|}{\tr{\rho P_{H(1)}} },
		\end{align}
		where in the last step we used that $\| R \rho R^\dag \|_1 = \tr{R \rho R^\dag}$. 
		
		Using Lemma~\ref{lem:perturbingprob}, we have seen that inequality~\eqref{eq:probbound1} holds. 
		We will apply Lemma~\ref{lem:perturbingprob} again.  Inverting the RHS inequality one gets
        \begin{align}
            \mathrm{Tr}(\rho A^\dag A) \geq \mathrm{Tr} (\rho \tilde{A}^\dag \tilde{A})
 - 2 \delta \|A\| - \delta^2.
        \end{align} 
        Setting $A= P_{H(1)}$, $\tilde{A} = R$, $\delta = \epsilon_P$ this gives
		\begin{align*}
			\tr{\rho P_{H(1)}} \geq \tr{R \rho R^\dag} - 2\epsilon_P - \epsilon_P^2, 
		\end{align*}
        
		and so $| \tr{R \rho R^\dag} -	\tr{\rho P_{H(1)}}| \leq 2\epsilon_P + \epsilon_P^2$.
		
		Putting things together
		\begin{equation}
			\left\| \frac{R \rho R^\dag}{\tr{ R \rho R^\dag}} -  \frac{R \rho R^\dag}{\tr{\rho P_{H(1)}}} \right\|_1  \leq 4 \epsilon_P + 2\epsilon_P^2.
		\end{equation}
		We also have that
		\begin{align} 
			\nonumber \left\|\frac{R \rho R^\dagger}{ \tr{\rho P_{H(1)}}} -  \frac{P_{H(1)} \rho P_{H(1)}}{\tr{\rho P_{H(1)}}}\right\|_1  &=\frac{\|R \rho R^\dagger - P_{H(1)} \rho P_{H(1)}\|_1}{ \tr{\rho P_{H(1)}}}  \\ & \leq 4\epsilon_{P}+2\epsilon_P^2.
		\end{align}
		Summing the contributions we conclude:
		\begin{equation}
			\left\| \frac{R \rho R^\dagger}{\tr{ R \rho R^\dag}} - \frac{P_{H(1)} \rho P_{H(1)}^\dag}{\tr{ P_{H(1)} \rho}} \right\|_1  \leq 8\epsilon_P + 4\epsilon_P^2.
		\end{equation}
        From this, using Eq.~\eqref{eq:equivalentform} and contractivity of the trace norm under the partial trace:
        	\begin{equation}
			\left\| \tilde{\sigma} - \ketbra{y}{y}\right\| \leq 8\epsilon_P + 4\epsilon_P^2,
		\end{equation}
        where $\tilde{\sigma} = \mathrm{Tr}_{1,2,3}\left[\frac{R \rho R^\dagger}{\tr{ R \rho R^\dag}}\right]$.
	\end{proof}
    We set $8\epsilon_P + 4\epsilon_P^2 = \epsilon$, which gives 
\begin{align}
	    \epsilon_P = \sqrt{1 + \epsilon/4} -1.
	\end{align}
    Using this for the filtering step with the gap $\Delta = 1/\kappa$, and $\alpha_{H(1)}=\alpha$, and adding to the expected adiabatic query count upper bound of $ 835.4\alpha \kappa$, gives a total expected count on success of
    \begin{align}
      Q^* \le   835.4 \alpha \kappa + l_P &=  835.4 \alpha \kappa + \left \lceil \alpha \kappa \ln \frac{2}{\sqrt{1+\epsilon/4}-1} +2 \right \rceil.
    \end{align}
    The success probability is,
    \begin{align}
        p_{succ} &\ge \frac{1}{2} -2 (\sqrt{1+\epsilon/4}-1) \nonumber \\
        & \geq 1/2 - \epsilon/4
    \end{align}
    and so the total expected cost is
    \begin{equation}
        Q \leq \frac{2Q^*}{1- \epsilon/2}.
    \end{equation}

 Note in the expression for $Q^*$ that if we set $\alpha = 1$, $\kappa = 10^6$ and $\epsilon = 10^{-10}$ we have an adiabatic cost of $\sim 8 \times 10^8$, whereas the cost of filtering is $\sim 1.4 \times 10^6$, so the former dominates. For $\kappa =10^3$ these two costs become $835331$ and $1387$, respectively. 

 \section{Summary and outlook}

 We have developed a quantum linear solver algorithm whose complexity scales optimally in the condition number (as $O(\kappa)$) and in the error (as O($\log 1/\epsilon$)). We have also provided an extensive analysis giving guaranteed upper bounds on the worst-case non-asymptotic query counts. 
 
Our algorithm is constructed combining a range of  techniques. We substantially improve on the adiabatic computing-inspired algorithm given in~\cite{subacsi2019quantum}, via modifications that are informed by eigenpath traversal theory~\cite{boixo2009eigenpath, sanders2020compilation}, quantum eigenstate filtering~\cite{lin2020optimal, costa2022optimal}  (for which we give an expanded analysis), and improved block-encoding construction. We also incorporate a Poissonization technique introduced in~\cite{cunningham2024eigenpath}, which we optimize and for which we present a tighter analysis based on adiabatic approximation theory~\cite{jansen2007bounds}. Finally, we introduce a randomized walk operator method into the quantum linear solver, inspired by works on eigenphase traversal via unitaries~\cite{boixo2009eigenpath} and previous works where the need for Hamiltonian simulation was removed from Quantum Phase Estimation~\cite{poulin2018quantum}. This allows us to entirely forego the quantum Hamiltonian simulation subroutine and correspondingly the need for classical phase factor precomputations, which considerably simplifies the algorithmic implementation and compilation. 

Our analysis returns, for a non-Hermitian linear system, a  query query count upper bound at $\epsilon=10^{-10}$ of $1\,722\kappa$. This is  intermediate between the two other guaranteed worst-case non-asymptotic query counts upper bounds currently available, which have $ 2 34 \, 562 \kappa$~\cite{costa2022optimal} and $80 \kappa$~\cite{dalzell2024shortcut}, as reported in the latter work. These numbers are important in setting a ceiling to the worst-case costs, but they do not define a ranking between the algorithms. We could complement these studies with numerical explorations,  but these are necessarily restricted to low-dimensional problem instances and small condition number, e.g.  $N \leq 16$, $\kappa \leq 50$ in Ref.~\cite{costa2023discrete}. 

In terms of further improvements to the algorithm, our choice of sampling the walk operator is certainly not optimal, even if we require perfect dephasing of the nullspace from other eigenspaces.  It would be of interest, however, to drop the condition of perfect dephasing and allow additional single-step errors, since cancellations can occur for multiple steps, as shown in~\cite{chiang2014improved}. Leveraging these results could provide non-trivial constant prefactor gains.

Another route is to improve query count bounds by exploiting structure about the problem. For example, further information about the distribution of singular values in the linear system matrix would give us a better handle on the eigenvalues of the associated $H(s)$; in turn, this would allow us to sharpen the Poissonization analysis. It would also be of value to explore how the linear-solver algorithm can be tailored to important sub-classes of  problems, for example in the context of linear systems arising from discretization of partial differential equations. We leave these directions to future work.

 	\bigskip
	
		{\textbf{Authors contributions:} Authors are listed alphabetically. ML conceived the core algorithm, following discussions with YS on variable-time amplitude amplification and adiabatic methods.
			ML optimized the randomized method, aided by analytical analysis from YS.
            YS and ML introduced the randomized walk operator method.
			DJ and ML developed the filtering component and the error propagation analysis.
			ML, DJ and YS computed the analytical cost of the algorithm.
			ML and DJ wrote the paper and AS, YS, SP contributed to reviewing the article.
			SP and AS coordinated the collaboration.
	
	\bigskip
	
	{\textbf{Acknowledgements:} Special thanks to Robert  B Lowrie, Sukin Sim and Dong An for insightful suggestions, Jessica Lemieux for comments on an earlier draft, William Pol for introducing us to the idea of QSP `multiplexing', Dominic Berry for useful clarifications concerning the constant prefactors in Ref.~\cite{costa2022optimal}, Tyler Volkoff for help with technical details in Lemma~\ref{lem:seconderivariveP}, Stephan Eidenbenz for scientific discussions and help in coordinating this collaboration. Thanks to all the colleagues at PsiQuantum for useful discussions and support. ML acknowledges the kind hospitality from the group of Rosario Fazio at the International Center for Theoretical Physics (ICTP) in Trieste, where part of this work was carried out. A.T.S. and Y.S. acknowledge support from US Department of Energy, Advanced Simulation and Computing Beyond Moore’s Law program. Y.S. acknowledges support from US Department of Energy, Office of Science, Office of Advanced Scientific Computing Research, Accelerated Research in Quantum Computing program. }
 
	\bibliography{Bibliography}

    \appendix

   \end{document}